\documentclass[conference]{IEEEtran}

\usepackage{amsmath,amsfonts, amsthm}
\usepackage{algorithmic}
\usepackage{graphicx}
\usepackage{textcomp}
\usepackage{xcolor}
\usepackage{hyperref}
\def\BibTeX{{\rm B\kern-.05em{\sc i\kern-.025em b}\kern-.08em
    T\kern-.1667em\lower.7ex\hbox{E}\kern-.125emX}}
\usepackage[skip=0.5em]{caption}
\usepackage{subcaption}
\usepackage{cleveref}
\usepackage{algorithm}
\usepackage[algoruled,algo2e, noend]{algorithm2e}
\usepackage{comment} 
\usepackage{multirow}
\usepackage{cancel}

\newcommand{\Source}{\mathcal{S}}
\newcommand{\Target}{\mathcal{T}}

\newcommand{\Algo}{\mathcal{A}}
\newcommand{\Randomizer}{\mathcal{R}}
\newcommand{\Output}{\mathcal{O}}
\newcommand{\Dataset}{\mathcal{D}}

\newcommand{\Plie}{\sigma}
\newcommand{\NotPlie}{\Bar{\sigma}}
\newcommand{\truth}{\NotPlie}
\newcommand{\lie}{\left(\frac{\Plie}{T-1}\right)}
\newcommand{\tinylie}{\frac{\Plie}{T-1}}

\newcommand{\I}{\mathbb{I}}

\newcommand{\Div}{\mathbb{D}}

\newtheorem{theorem}{Theorem}
\newtheorem{lemma}{Lemma}
\newtheorem{example}{Example}
\newtheorem{definition}{Definition}

\newtheorem{remark}{Remark}

\usepackage[disable]{todonotes} 
\newcommand{\aurelien}[1]{\todo[inline,caption={},color=green!40]{{\it Aurelien:~}#1}}
\newcommand{\riad}[1]{\todo[inline,caption={},color=blue!40]{{\it Riad:~}#1}}
\newcommand{\nicolas}[1]{\todo[inline,caption={},color=orange!40]{{\it Nicolas:~}#1}}
\newcommand{\gs}[1]{\todo[inline,caption={},color=purple!40]{{\it Guillaume:~}#1}}

\renewcommand{\arraystretch}{0.3} 
\newcommand{\Inn}[6]{\!\!
\left(
    \begin{tabular}{r|l}
        \!\tiny $#1$ & \tiny $#4$\!\\
        \!\tiny $#2$ & \tiny $#5$\!\\
        \!\tiny $#3$ & \tiny $#6$\!
    \end{tabular}
\right)
}

\newcommand{\Ink}[3]{
\left(
    \begin{tabular}{c}
        \tiny $#1$ \\
        \tiny $#2$ \\
        \tiny $#3$ 
    \end{tabular}
\right)
}

\newcommand{\PNDum}[6]{ \mathbb{P}_{n,d}^{\text{dum}} \Inn{#1}{#2}{#3}{#4}{#5}{#6} }

\newcommand{\PND}[6]{\mathbb{P}_{n,d} \Inn{#1}{#2}{#3}{#4}{#5}{#6} }

\newcommand{\PNZ}[6]{\mathbb{P}_{n,0} \Inn{#1}{#2}{#3}{#4}{#5}{#6} }

\usepackage{array}
\newcommand{\PreserveBackslash}[1]{\let\temp=\\#1\let\\=\temp}
\newcolumntype{C}[1]{>{\PreserveBackslash\centering}m{#1}}
\newcolumntype{R}[1]{>{\PreserveBackslash\raggedleft}m{#1}}
\newcolumntype{L}[1]{>{\PreserveBackslash\raggedright}m{#1}}

\makeatletter
\newcommand*\rel@kern[1]{\kern#1\dimexpr\macc@kerna}
\newcommand*\widebar[1]{%
  \begingroup
  \def\mathaccent##1##2{%
    \rel@kern{0.8}%
    \overline{\rel@kern{-0.8}\macc@nucleus\rel@kern{0.2}}%
    \rel@kern{-0.2}%
  }%
  \macc@depth\@ne
  \let\math@bgroup\@empty \let\math@egroup\macc@set@skewchar
  \mathsurround\z@ \frozen@everymath{\mathgroup\macc@group\relax}%
  \macc@set@skewchar\relax
  \let\mathaccentV\macc@nested@a
  \macc@nested@a\relax111{#1}%
  \endgroup
}
\makeatother

\newcommand{\Abar}{\widebar{\Algo}}
\newcommand{\Sbar}{\widebar{\mathcal{S}}}

\newif\ifpreprint
\preprinttrue
\newcommand{\inPreprint}[2]{\ifpreprint #1\else #2\fi}
\begin{document}
\title{Mitigating Leakage from Data Dependent Communications in Decentralized Computing using Differential Privacy}


\author{\IEEEauthorblockN{Riad Ladjel\IEEEauthorrefmark{1}, Nicolas Anciaux\IEEEauthorrefmark{1}, Aurélien Bellet\IEEEauthorrefmark{2}, Guillaume Scerri\IEEEauthorrefmark{1}}
\IEEEauthorblockA{
\IEEEauthorrefmark{1}Petrus team, Inria, France
\IEEEauthorrefmark{2}Magnet team, Inria, France}}

\maketitle
\begin{abstract}

  Imagine a group of citizens willing to collectively contribute their personal
  data for the common good to produce socially useful information, resulting
  from data analytics or machine learning computations. Sharing raw personal
  data with a centralized server performing the computation could raise concerns
  about privacy and a perceived risk of mass surveillance. Instead, citizens may
  trust each other and their own devices to engage into a decentralized
  computation to collaboratively produce an aggregate data release to be shared.
  In the context of secure computing nodes exchanging messages over secure
  channels at runtime, a key security issue is to protect against external
  attackers observing the traffic, whose dependence on data may reveal personal
  information. Existing solutions are designed for the cloud setting, with the goal
  of hiding all properties of the underlying dataset, and do not address the
  specific privacy and efficiency challenges that arise in the above context. In
  this paper, we define a general execution model to control the data-dependence
  of communications in user-side decentralized computations, in which
  differential privacy guarantees for communication patterns in global execution
  plans can be analyzed by combining guarantees obtained on local clusters of
  nodes. We propose a set of algorithms which allow to trade-off between
  privacy, utility and efficiency. Our formal privacy guarantees leverage and
  extend recent results on privacy amplification by shuffling. We illustrate the
  usefulness of our proposal on two representative examples of decentralized
  execution plans with data-dependent communications.

\end{abstract}






\pagestyle{plain}


\section{Introduction} 
\label{sec:intro}
  
\aurelien{introduire un peu de structure explicite dans l'intro serait pas mal je pense; genre "Problem statement"; "Related work"; "Contributions", en utilisant des backslash paragraph \nicolas{J'avoue que les intro découpées de cette facon me plaisent en général un peu moins, ca risque de couper le fil de l'histoire racontée. Aurélien si tu as une idée plus précise, je te laisse bien sûr carte blanche. De mon côté, vus les délais je me concentrerai sur le running ex et la section 5 (qui reste encore à adapter à la nouvelle section 2, et avec le pb de l'évaluation de la privacy sur l'arbre Agregate à rectifier)}}

In many areas such as health, social networks or smart cities, the need to obtain information from citizens for the social good is growing. Hospitals are asking patient groups to provide health statistics\footnote{See e.g., the ComPaRe initiative in Paris hospitals: \url{ https://compare.aphp.fr/}} related to diseases for which little data is available or about recently recognized pathologies \cite{robert2021production}. Researchers also ask citizen volunteers to collectively use their social media accounts to identify potential abuse and political propaganda \cite{silva:hal-03048337}. Similarly, municipalities are using participatory sensing applications to collect urban statistics (e.g., on street noise exposure \cite{hachem2015monitoring, popa2019mobile, brahem2021consent}).
While legal frameworks such as data portability \cite{EUGDPR2016, pardau2018california} or data altruism \cite{EUDataGovAct2020} allow citizens to retrieve and share their personal data, asking citizens to provide raw data raises privacy concerns, with a risk of being perceived as mass surveillance \cite{wagner2021pandemica} and limited chances for widespread adoption~\cite{kramer2020making}. 

Instead, in this work we consider a setting where groups of citizens wish to collectively compute and share an \emph{aggregate data release}, resulting from a distributed database query or a federated machine learning process \cite{kairouz2019advances} executed on their personal data. 
In terms of security, this implies considering a distributed set of trusted computing nodes processing data locally, and exchanging results at runtime through secure communication channels. While techniques aiming to ensure that a computation does not leak information (including through communication patterns) exist, from secure multiparty computations \cite{Bonawitz17} to secure outsourcing \cite{Chowdhury20}, 
they either have a very high overhead when massive number of parties are involved or rely on distributing trust between a small number of servers which is contradictory with our massively distributed setting. In order to avoid such drawbacks, we thus avoid changing the way distributed computations are performed and wish to keep computations at the level of client nodes.
Securing local processing from an attacker 
can be tackled by solutions ranging from software security~\cite{bater2017smcql,MPC:MPCgoneLive} to hardware-based protection~\cite{LMR,Ryoan, sh:oblivious2}. These solutions suggest that avoiding secure servers and relying only on simple client nodes is a feasible and sound approach.

In this paper, we focus on the problem of protecting against traffic analysis between computing nodes, whose dependence on data may reveal sensitive (e.g., personal) information. This is a standard, independent security problem that has been less studied than leakage from other side channels (such as memory access patterns~\cite{zheng2017opaque, xu2019hermetic,DP-Orimenkho,DP-Mazloom}), especially in the context of massively decentralized processing. 
Data-dependent communications in execution plans of distributed queries arise for
efficiency reasons (to enable nodes to do more work in parallel) and depend on the query to be evaluated. For example, in distributed SQL and MapReduce analytics, nodes process data based on given (group-by) value intervals or (join) hashed key values (see e.g., \cite{bellamkonda2013adaptive}). In distributed data mining and machine learning algorithms, iterative updates to the models are split across computing nodes based on the distance to given centroids, according to regions of the feature space, or based on the similarity of local updates \cite{zhao2009parallel,kmedoids,RP-DBSCAN,clustfl1,clustfl2}.

Existing solutions to hide the dependency between communication patterns and underlying data values have been proposed for the so-called \emph{confidential computing} model \cite{russinovich2021toward, sturzenegger2020confidential, CCCWhitePaper, CCCWhitePaper2, rashid2020rise}, where computing nodes are allocated in cloud settings within secure enclaves and generate data exchanges at query time over secure communication channels.
A first option is to make all communications (computationally) data independent. 
As a more practical alternative to the naive option of broadcasting dummy messages to all nodes, \cite{sh:observing} proposes to use communication padding (data transmitted from node to node is padded to a maximum size) and clipping (when the maximum size is reached, messages are discarded) for MapReduce computations. 
However, to avoid high communication overheads and preserve the utility of the computation, appropriately tuning the padding and clipping parameters is crucial. This requires prior knowledge on the data distribution to obtain a good balance between computing nodes, which is unrealistic in our context with potentially unknown sets of volunteers holding small data sets (in extreme cases, citizens contribute with a single personal record).

A second option is to anonymize communication patterns by resorting to secure shuffling \cite{Prochlo}, mixnets \cite{M2R} or oblivious data exchanges \cite{zheng2017opaque}. 
However, the privacy guarantees obtained in a massively decentralized setting are difficult to formally quantify, especially against attackers with auxiliary knowledge. Even without such knowledge, anonymized communication patterns may still leak sensitive information (e.g., observing communication patterns to computing nodes that collect personal data for a given range of values would reveal the number of involved citizens with values in this range).

Hence, hiding communication patterns in the case of massively decentralized computations on end-user nodes 
calls for new solutions.
A key challenge is to design techniques that can appropriately trade-off between privacy, utility and efficiency.

In this work, we propose to mitigate the leakage from data-dependent communication patterns with the tools of differential privacy, a mathematical definition of privacy with many interesting properties (including robustness against auxiliary knowledge). 
Our first contribution is to define a general execution model, in which the differential privacy guarantees of communications patterns resulting from global execution plans can be analyzed by combining guarantees obtained for local clusters of nodes through composition.
Our second contribution is a set of algorithms to hide cluster-level communication patterns. Our algorithms rely on a combination of local sampling (each node randomizing the targets of messages), flooding (adding dummy messages) and shuffling subsets of messages with scramblers. We prove analytical differential privacy guarantees for our solutions by relying and extending recent results on privacy amplification by shuffling \cite{Cheu2019,amp_shuffling,Balle2019}. 
Finally, our last contribution is to highlight the practical usefulness of our solutions on two representative examples of execution plans with data-dependent communication, showing possible privacy-utility-efficiency trade-offs.

The rest of the paper is organized as follows. Section~\ref{sec:problem} defines the problem. 
Section~\ref{sec:DP-data} introduces a first approach where each node locally randomizes its messages. 
In Section~\ref{sec:Privacy-anal}, we propose algorithms where the noise can be shared across nodes.
Section~\ref{sec:eval} presents an evaluation of our approach on two concrete use-cases. We review some related work in Section~\ref{sec:rw}, and conclude in Section~\ref{sec:conclusion}.

\section{Problem Definition} 
\label{sec:problem}

In this section, we describe our execution
and adversary models, 
formulate the problem of mitigating the dependency of communication to data when executing distributed queries, 
and present
the performance metrics used to evaluate the quality of solutions. 
We begin with a motivating example used throughout the section to illustrate the various concepts.


\subsection{Motivating Example} 
\label{sec:usecase}

We introduce here a representative example of distributed query producing aggregated values of certain attributes grouped by other sets of attributes for a dataset contributed by a large set of participants. This type of queries, called \emph{Aggregate} in the paper, is used to understand frequency distributions and collect marginal statistics from the data, as routinely performed in a data exploration phase. Such queries are also used in many data mining/machine learning algorithms (e.g., to learn decision trees). The data, computation and distribution models are as follows.

\paragraph{Data model} Personal data is hosted in users' source nodes and forms a partition of the dataset under consideration, as in federated machine learning or federated database scenarios. For simplicity, the dataset is represented by a single table $\mathcal{D}(\mathbb{A})$, with $\mathbb{A}$ a fixed set of attributes used in the computation, 
and any source node hosts a single tuple $t$ in $\mathcal{D}$.

\paragraph{Computation model} An Aggregate query on the dataset $\mathcal{D}$ is expressed as $Q=\{q\}$ a set of sub-queries, where each sub-query $q = (\mathcal{G}_q, \mathcal{V}_q, \mathcal{F}_q)$ is defined by a set $\mathcal{F}_q$ of statistical functions (e.g., \{count, avg, max\}) evaluated over a set of aggregate attributes $\mathcal{V}_q \subset \mathbb{A}$ and grouped by another set of attributes $\mathcal{G}_q \subset \mathbb{A}$. Note that $Q$ is a typical case of a SQL aggregate database query with a \emph{grouping sets} clause. 

\paragraph{Distribution model} For evaluating query $Q$, the dataset $\mathcal{D}$ is partitioned vertically by sub-query $q$, and horizontally on grouping attribute values of $\mathcal{G}_q$ for each sub-query. Each partition $\mathcal{D}_q^i$ of the dataset is then assigned to a distinct compute node $c_q^i$ in charge of computing $q(\mathcal{D}_q^i)$. The partitions $\mathcal{D}_q^i$ can be formed by locally (i.e., at source node) assigning any tuple $t \in \mathcal{D}$ to the appropriate partitions. The result of $Q = \cup_{q,i}\{q(\mathcal{D}_q^i)\}$ is obtained by the union of the results produced by all compute nodes $\{c_q^i\}$ on their partition $\mathcal{D}_q^i$.

For illustration purposes, in the rest of this section we consider the following query instance as a running example.
\begin{example}[Running example, Figure~\ref{fig:QEPex3}]
  \label{ex:use-case-full}
A dataset for a diabetes study consists of a table $\mathcal{D}(A,B,I)$ with two grouping columns, Age range ($A$) and Body mass index ($B$), and one aggregation column, Insulin rate ($I$). Each row of the table is a tuple $t_k=(a,b,i)$ with the values taken for a contributing patient.
A query $Q$ is the union of $q_1 = (A, I, avg)$ and $q_2 = (B, I, avg)$ which compute respectively the average insulin rate grouped by age range and grouped by body mass index. Note that this query is equivalent to the SQL database query: \it{SELECT A,B,AVG(I) FROM D GROUP BY GROUPING SETS (A),(B)}.
The dataset $\mathcal{D} = \{t_k\}_{k \in [1,d]}$ is
  distributed on a set of source nodes each holding one tuple
  $t_k$. To evaluate the query, a first vertical partition of the dataset
  $\mathcal{D}_1=\mathcal{D}(A,I)$ is used to evaluate $q_1$ and a second
  $\mathcal{D}_2=\mathcal{D}(B,I)$ to evaluate $q_2$. The grouping domains $A$ for $q_1$ and
  $B$ for $q_2$ are respectively partitioned horizontally into $C$ parts
  $\{A_i\}_{i \in [1,C]}$ and $\{B_j\}_{j \in [1,C]}$, with
  $\mathcal{D}_1^i = \{t \in \mathcal{D}_1, t.a \in A_i\}$ and $\mathcal{D}_2^j = \{t \in \mathcal{D}_2, t.b \in
  B_j\}$. The query execution is distributed on two groups $\mathcal{C}_1$ and $\mathcal{C}_2$ of $C$ compute
  nodes each, computing $q_1(\mathcal{D}_1)$ and $q_2(\mathcal{D}_2)$ respectively. A source node
  holding $t_k=(a_i,b_j,i_k)$ with $a_i\in A_i$ and $b_j\in B_j$ sends $(a_i,i_k)$ to compute node $c_1^i \in \mathcal{C}_1$
  which evaluates $q_1(\mathcal{D}_1^i)$ and $(b_j,i_k)$ to
  $c_2^j \in \mathcal{C}_2$
  which evaluates $q_2(\mathcal{D}_2^j)$. The result $R$ of the query
  is the union of the results produced by compute nodes.
\end{example}

\begin{figure}
\centering
    \includegraphics[width=.5\textwidth]{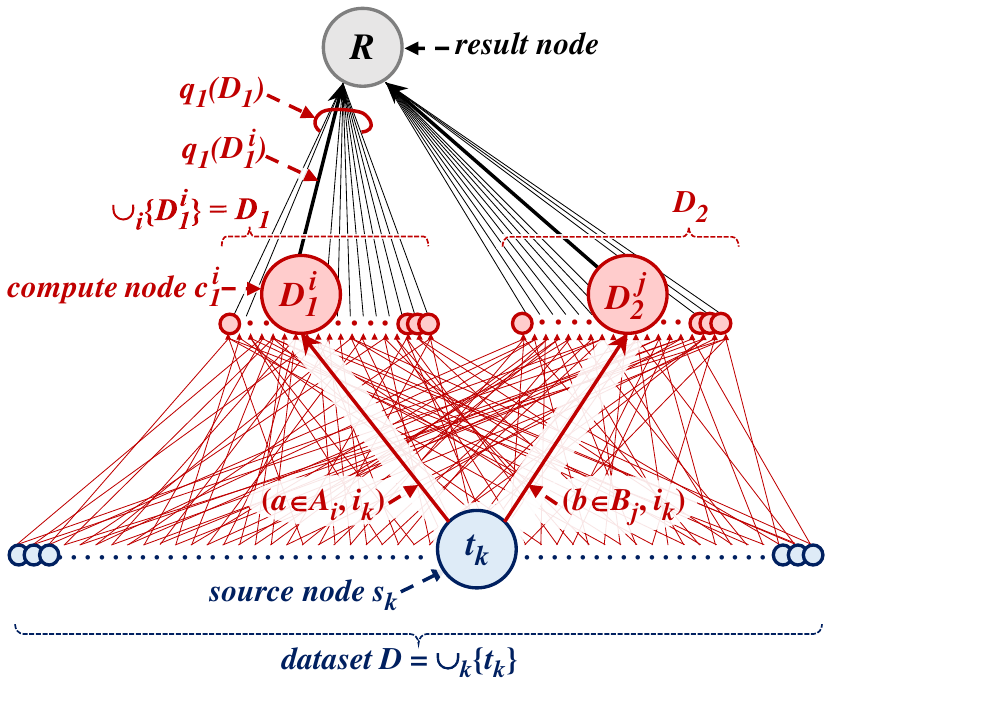}
    \caption{Distributed query plan (Example~\ref{ex:use-case-full}).}
    \label{fig:QEPex3}
    \vspace{-1em}
\end{figure}

\emph{Data leakage from communications.} In general, an attacker observing the destination of (encrypted)
messages sent by a user contributing to the computation 
deduces the values of the grouping attributes
$\mathcal{G}_1,\dots,\mathcal{G}_n$ for this user. For example, the user contributing tuple
$(A=19,B=18,I=128)$ in Example~\ref{ex:use-case-full} 
sends $(A=19,I=128)$ to the compute node 
processing $A=19$, and $(B=18,I=128)$ to the one processing $B=18$, and thus
reveals their age and body mass index to the attacker. By extension, an attacker observing a
compute node identifies users sharing the values processed by that compute
node. If some compute nodes process sensitive values (e.g., nodes corresponding
to low or high BMI values in Example~\ref{ex:use-case-full}), the communications reveal potentially
discriminated users. More generally, the higher the number of grouping sets (desirable in terms of utility) and the finer the partitioning (beneficial
from a performance viewpoint), the more extensive and accurate is the information
inferred about contributors. Ultimately, an attacker observing the messages in
the computation could infer all the grouping values of the participants, which
amounts to revealing their profile (if not their identity, as a few demographic
attributes are enough to uniquely identify an individual
\cite{DBLP:conf/www/RocherMM21}).

\subsection{Execution Model}


Our execution model aims to abstract away the specific computation to focus on
capturing generic communication patterns. In order to achieve this, we break
down the distributed computation between a set of elementary \emph{nodes} that
execute arbitrary computations but have simple communication behavior. We assume
that nodes communicate on secure channels enforcing integrity and secrecy of
communications (e.g. TLS).
We then combine these nodes in order to obtain an \emph{execution plan}.

\paragraph{Nodes}
The simplest building block of our execution model is a node, which executes an
elementary step of the computation. We assume that each node has a unique
unforgeable public identifier (typically provided by a PKI).
The input of this elementary computation consists of data received from other
nodes and/or internal data of the node (represented as a set of tuples). As we
focus on communication between nodes, we allow nodes to run any program on
received data but explicitly require
that 
a node sequentially follows three steps:
\begin{enumerate}
\item Receive a number of input messages from a set of source nodes (if any);
\item Process input data together with internal data of the node;\footnote{We
  do not consider the leakage that may occur during computation due to side channels (e.g. timing), 
  see Section~\ref{sec:rw} for a review of mitigation techniques based e.g. on constant time code, ORAMs, etc.}
\item Establish secure channels with target nodes and send a number of
  messages to a set of target nodes (if any).
\end{enumerate}
We assume that the number of messages received and sent by a given node is fixed
(if it is not, we add a dummy target node that represents discarding the
message). \nicolas{un dummy target convient au cas où le noeud retire des messages, mais quid s'il en rajoute? ex: join, produit cartésien...}
We further assume that the size of messages does not depend on the
value taken by the data (content of the message).

Given a node $\nu$ and its input data $\mathcal{D}_\nu$ (a set of tuples consisting of internal
data and messages received from other nodes), the execution of $\nu$ on $\mathcal{D}_\nu$ generates
some communication in the form of a set of triplets $(\nu, m_i, t_i)$ indicating
that node $\nu$ sent a message with content $m_i$ to target node $t_i$. We denote by $\nu(\mathcal{D}_\nu)$ the set of triplets produced by the execution of $\nu$ on $\mathcal{D}_\nu$. Since messages are indistinguishable to attackers
(as we use secure channels between nodes and assume same size messages), in
the rest of the paper we will often omit $m_i$.

\begin{example}
\label{ex:nodes}
  The computation described in Example~\ref{ex:use-case-full} consists of four types of nodes:
  \begin{itemize}
  \item A set of source nodes, each holding one data tuple $t_k=(a_i,b_j,i_k)$ where $a_i\in A_i$ and $b_j\in B_j$. They do not receive
    messages and send two messages: $(a_i,i_k)$ to the compute node $c_1^i$ in $\mathcal{C}_1$ processing $\mathcal{D}_1^i$, and
    $(b_j,i_k)$ to the compute node $c_2^j$ in $\mathcal{C}_2$ processing $\mathcal{D}_2^j$.
  \item Two types of compute nodes $\mathcal{C}_1 = \{c_1^i : 1\leq i\leq C\}$ and $\mathcal{C}_2 = \{c_2^j : 1\leq j\leq C\}$ which receive a number of messages from the source nodes and send out one message (the aggregated
    result).
  \item A result node $R$ which receives the (partially) aggregated data from the compute nodes.
  \end{itemize}
\end{example}

\gs{Peut-être donner un ou deux exemples ici. Group by MR me semble bien avec
  sending node, reducing node, aggregating node par ex.
  \aurelien{ce serait super bien. par exemple en s'appuyant sur une figure, potentiellement en 2 parties: 1/ execution plan avec juste des nodes (référencée ici), et 2/ variante avec les clusters (référencée plus loin). mais on manque sans doute de temps}}


\paragraph{Execution plans}
An execution plan is a set of nodes with disjoint node identifiers (if an ``agent''
needs to execute multiple computation steps, it executes several nodes). A
distributed execution is simply the execution of all nodes where each message is
taken as input to the specified target node.

\begin{definition}[Communication graph of an execution plan]
  Given a set of nodes $\mathcal{N}$, a  dataset 
  $\mathcal{D} = \bigcup_{\nu\in\mathcal{N}} \mathcal{D}_\nu$, and $r$ the randomness used in the
  computation (divided across all nodes), we define the communication graph
  $\mathcal{N}^r(\mathcal{D})$ of the execution plan $\mathcal{N}$ on data
  $\mathcal{D}$ with randomness $r$ as the union of the communication outputs
  $\{(\nu, t_i)\}=\bigcup_{\nu\in\mathcal{N}} \nu(\mathcal{D}_\nu)$ produced by each node $\nu\in\mathcal{N}$ with input data $\mathcal{D}_\nu$ (message contents are excluded as they are hidden by secure channels and have same size) during the computation, with edges labeled by
  the order number of the message sent. 
\end{definition}
Probabilities will be taken over $r$. When clear from the context, $r$ may be omitted.

\begin{example}
\label{ex:com_graph}
  In our example, the communication graph for the execution plan outlined in
  Example~\ref{ex:use-case-full} is the graph presented in
  Figure~\ref{fig:QEPex3} where a source node $s$ holding tuple $t_k=(a_i,b_j,i_k)$ with $a_i\in A_i$ and $b_j\in B_j$ has an edge
  labeled $1$ to the compute node $c_1^i\in\mathcal{C}_1$ processing $(a_i,i_k)$, and an
  edge labeled $2$ to the compute node $c_2^j\in\mathcal{C}_2$ processing $(b_j,i_k)$, and all compute nodes have an outgoing edge to the result node $R$.
\end{example}

\subsection{Adversary and Security Models}
\label{sec:adversary-model}


We assume that node identifiers cannot be forged by the adversary in order to
impersonate nodes. This could typically be achieved through the use of a public
key infrastructure (PKI). In the case of SGX for example, Intel would play the
role of such a PKI, ensuring that the computation environment is what we expect.
We assume that all nodes communicate on secure channels that provide two-way
authentication, secrecy of messages (i.e. an adversary cannot distinguish
between a true message and a random message of the same size), and integrity of
communications (i.e. an adversary may not convince the receiver that a forged
message was sent by the sender), under some cryptographic assumption
$\mathcal H$ (e.g. decisional Diffie–Hellman). For example, TLS with client authentication would
satisfy these conditions under the cryptographic hypotheses ensuring that the
adversary may not break the underlying encryption and signature schemes.

We consider an adversary observing all communications, and only communications,
in particular the adversary cannot observe the internal state of
nodes. Additionally we assume that the adversary cannot not break the cryptographic
hypotheses $\mathcal H$ (and is otherwise unbounded). Therefore, the adversary cannot
break the security of secure channels and thus cannot observe the content
of messages. 

The natural way to model the amount of information leakage of an execution plan
$\mathcal{N}$ is differential privacy, applied here to communication patterns.
Formally, given an execution plan $\mathcal{N}$ and a dataset $\mathcal{D}$, the
adversary is given access to a computation of $\mathcal{N}(\mathcal{D})$ and
tries to infer information about individual tuples in $\mathcal{D}$ (e.g., the
values of attributes $A$ and $B$ of a user in Example~\ref{ex:use-case-full})
leading to a natural definition similar to the computational differential
privacy of \cite{ComputationalDP}.  As usual with differential privacy, we do
not make any assumption on the auxiliary knowledge of the adversaries (they may
have arbitrary knowledge and even know some tuples of the dataset). We then require that the adversary should not be able to distinguish two runs of the protocol with neighboring datasets (i.e. datasets differing by exactly one tuple) with good probability. As the
adversary cannot break secure channels, we can abstract away the content of messages
and the construction of secure channels (for more details see
Appendix~\ref{app:comp-diff-priv}), giving only the communication graph as
observables for the adversary.  Formally, under hypotheses $\mathcal{H}$, we have the following privacy definition.


\begin{definition}[DP for execution plans]
\label{def:dp1}
  An execution plan $\mathcal{N}$ is $(\epsilon, \delta)$-differentially private
  if for any neighboring $\mathcal{D}_0, \mathcal{D}_1$ (differing by at most one
  tuple), and for all possible sets $\mathcal{O}$ of communication graphs
  , we have:
  \[
    P[\mathcal{N}(\mathcal{D}_0)\in\mathcal{O}] \leq e^\epsilon P[\mathcal{N}(\mathcal{D}_1)\in\mathcal{O}] + \delta.
  \]
\end{definition}

Note that providing privacy guarantees for releasing the \emph{result} of the
computation is an orthogonal problem, but one advantage of using differential
privacy for protecting communications is that it will compose well with techniques that provide
differential privacy guarantees for protecting the result or other side channels (see Section~\ref{sec:RW-side}).

\subsection{Reducing the Problem to Clusters of Nodes}
\label{sec:simplifying-problem}

In order to analyze the problem and control data dependency, we restrict
the type of nodes we consider as follows.

\begin{definition}[Simple node]
  A simple node is a node such that 
  modifying one input tuple in $\mathcal{D}$ may only change one output message of this node (content
  and target).
\end{definition}

Importantly, this does not restrict the type of distributed computations that we can do in
practice. For instance, although in Example~\ref{ex:nodes} a source node is not simple (two
messages depend on the tuple $t_k$), it can be seen as a pair of simple
``logical'' nodes running on the same client, one sending $(a,i)$, the other one
sending $(b,i)$. Note that when doing this transformation the input data of
nodes is no longer disjoint between nodes. These two ``logical'' nodes share the same
base identity (the identity of the physical node) and are differentiated by the
order number of the message they send.

As a preliminary step towards achieving differential privacy for communications
in an arbitrary execution plan, we reduce the problem to considering a set of simple
nodes with the same communication behavior, namely a \emph{cluster} of
nodes. This allows for a generic analysis as we will see in Section~\ref{sec:DP-data}.\footnote{In particular, for a cluster, the only data dependency is which target receives a message, and differential privacy of all communications essentially becomes recipient anonymity as defined in \cite{AnoA}.}
Moreover, working at the level of a cluster will allow us to share noise
addition for these nodes (see Section~\ref{sec:Privacy-anal}), obtaining better
guarantees with less traffic 
while reasoning locally
on the potential communications. Crucially, we show that if we provide
privacy guarantees for each cluster of an execution plan, they translate into
privacy guarantees for the overall execution plan by composition (see
Theorem~\ref{th:compo-clusters}).

\begin{definition}[Cluster of nodes]
  In an execution plan $\mathcal{N}$, a set of nodes $\mathcal{C} \subseteq \mathcal{N}$ is a cluster if
  \begin{itemize}
  \item the set of potential target nodes $\mathcal{T}\subseteq \mathcal{N}$ for each message sent by any
    node in $\mathcal{C}$ is the same,
  \item all nodes in $\mathcal{C}$ are simple,
  \item nodes in $\mathcal{C}$ operate on disjoint data.  
\end{itemize}
\end{definition}

\begin{example}
\label{ex:clusters}
  In our running example, source nodes (users contributing data) are not \emph{simple} as they produce two messages (see Figure~\ref{fig:QEPex3}) and thus cannot directly be divided into clusters. However, breaking down each source node $s$ into two logical nodes $s_1$ and $s_2$ sending out messages for the first (resp. second)
  set of target nodes $\mathcal{C}_1$ (resp. $\mathcal{C}_2$), all nodes are simple. Denoting by $\mathcal{S}_1$ (resp. $\mathcal{S}_2$) the nodes communicating with $\mathcal{C}_1$ (resp. $\mathcal{C}_2$), note that the nodes in $\mathcal{S}_1$ operate on disjoint data (different tuples of $\mathcal{D}$), as
  do nodes in $\mathcal{S}_2$ (but nodes in $\mathcal{S}_1$ share data with nodes in $\mathcal{S}_2$, as nodes $s_1$ and $s_2$ process the same tuple $t \in \mathcal{D}$). The
  computation can therefore be divided into five clusters: $\mathcal{S}_1$, $\mathcal{S}_2$, $\mathcal{C}_1$
  and $\mathcal{C}_2$ and $R$, as shown in Figure~\ref{fig:Clusters-ex3}. Note that communications originating from $\mathcal{C}_1$ and $\mathcal{C}_2$
  are not data dependent (they only send out one message to $R$), hence they do not need any countermeasure to hide their
  communication patterns. $R$ does not send messages at all, and will therefore
  be ignored in the remainder of the paper.
\end{example}

For simplicity and without loss of generality, in the rest of the paper we assume  that nodes in a cluster only send one message.
Indeed, as each input tuple of a simple node may only change one output message, all output messages can be computed independently.


As our goal is to modify execution plans without altering the underlying
computation, we are restricted to working on communication graphs rather than
the internal workings of nodes. Therefore, all our algorithms will take
communication graphs as input and return new, perturbed communication graphs.
With a reduction similar to the one presented in \cite{AnoA}, based on our previous assumptions, we can abstract away
the dependence on data and provide the following differential privacy
definition, which is equivalent to Definition~\ref{def:dp1} for clusters of
nodes. Indeed, with our definition of simple nodes, two neighboring datasets may only produce
communication patterns differing by the recipient of at most one message, defining \emph{neighboring communication graphs}.\footnote{This is in line with the notion of adjacency for recipient anonymity in \cite{AnoA}.}

\begin{definition}[DP for communication graphs]
  \label{def:dp2}
  An algorithm $\Algo$ is $(\epsilon,\delta)$-differentially private if for any two neighboring
  communication graphs $\mathcal{G}_1, \mathcal{G}_2$ (differing in at most one message) and any set of communication graphs $\Output$, we have:
\[
  P[\Algo(\mathcal{G}_1) \in \Output]  \leq e^{\epsilon} P[\Algo(\mathcal{G}_2) \in \Output)] + \delta.
\]
\end{definition}
If $\Algo$ can be written as a transformation of an
execution plan $\mathcal{N}$ (potentially by adding nodes in the middle of the original communication paths),
with a slight abuse of notations we denote this transformation as $\Algo(\mathcal{N})$.




Finally, as execution plans may be composed of multiple clusters, we need a way
of combining the privacy guarantees of each cluster into a guarantee for the overall
execution plan. In order to do this, we take advantage of the composition
properties of differential privacy and use the simple observation that a given data tuple may only
influence communications along the communication path originating from its source.

\begin{theorem}
  \label{th:compo-clusters}
  Let $\mathcal{N}$ be a valid execution plan where vertices form a set $I$ of
  disjoint clusters $(\mathcal{N}_i)_{i\in I}$. If for each $i\in I$,
  $\Algo(\mathcal{N}_i)$ is $(\epsilon_i, \delta_i)$-differentially private, then
  $\Algo(\mathcal{N})$ is $(\epsilon, \delta)$-differentially private where $\epsilon$
  and $\delta$ are the worst possible sums of $\epsilon_i$'s and $\delta_i$'s encountered
  along the path of a data item the communication graph:
  \begin{align*}
    \epsilon &= \max_{\text{path }p\in\mathcal{N}}~\sum_{i : \exists \nu\in\mathcal{N}_i \text{ s.t. } n\in p}\epsilon_i,\\ 
    \delta &= \max_{\text{path }p\in\mathcal{N}}~\sum_{i : \exists \nu\in\mathcal{N}_i \text{ s.t. } n\in p}\delta_i.
  \end{align*}
\end{theorem}
\begin{proof}
The result follows from applying a combination of classical sequential and parallel composition results 
for differential privacy, see \cite{Dwork2014a}.
\end{proof}

\begin{remark}
Execution plans typically require a small number of clusters. For instance, a single cluster is sufficient for simple Aggregate queries (with one grouping set).
For computations that require an iterative process (e.g., K-means and machine
learning algorithms in general), we make the graph acyclic by considering that each iteration is done
by a different set of nodes (and each ``agent'' will execute several nodes). In
terms of clusters, this means that each iteration adds more clusters and the privacy
guarantee we obtain degrades with the number of iterations, as usual when
considering differential privacy for this type of computation. Some
optimizations can be made for specific computations, but we leave this for future
work.
\end{remark}

\begin{figure}
\centering
    \includegraphics[width=.4\textwidth]{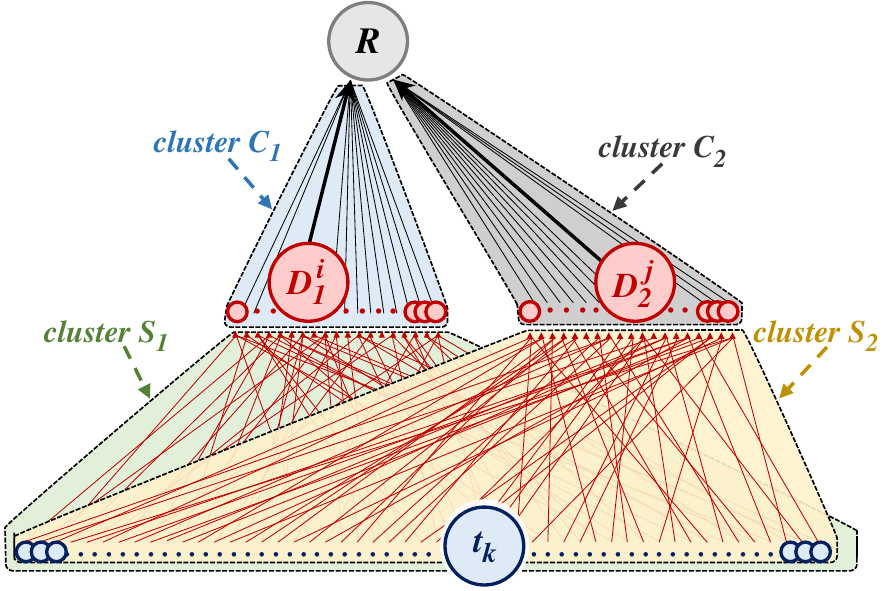}
    \caption{Clusters of nodes (Example~\ref{ex:clusters}).}
    \label{fig:Clusters-ex3}
    \vspace{-1em}
\end{figure}

\begin{example}
  In our running example, each data item encounters clusters $\mathcal{S}_1,\mathcal{S}_2,\mathcal{C}_1,\mathcal{C}_2$
  along its path. As mentioned before, $\mathcal{C}_1$ and $\mathcal{C}_2$ have data-independent
  communication outputs, and thus provide perfect privacy and do not need
  countermeasures. Therefore, if $\Algo$ provides $(\epsilon_1, \delta_1)$ and $(\epsilon_2, \delta_2)$-DP for $\mathcal{C}_1$ and $\mathcal{C}_2$ respectively, then $\Algo(\mathcal{N})$
  is $(\epsilon_1+\epsilon_2, \delta_1+\delta_2)$-DP.
\end{example}

\begin{example}
  If one is willing to trade off utility for a gain in privacy, another possible way
  of implementing (an analogue of) this computation would be for each
  participant to decide whether they want to participate to the aggregate on $a$
  or the aggregate on $b$. In this new execution plan $\mathcal{N}'$, each
  data tuple goes through exactly two clusters along its path (either $\mathcal{S}_1,\mathcal{C}_1$ or $\mathcal{S}_2,\mathcal{C}_2$)
  and $\Algo(\mathcal{N}')$ is
  $(\max\{\epsilon_1, \epsilon_2\}, \max\{\delta_1, \delta_2\})$-DP.
\end{example}

As any execution plan can be broken down into clusters, and Theorem~\ref{th:compo-clusters} provides DP guarantees on the whole execution plan from guarantees on clusters, we focus on providing guarantees at the cluster level. Specifically, Section~\ref{sec:DP-data} takes advantage of the fact that all nodes in a cluster have similar communication behavior in order to derive a DP bound using local differential privacy, while Section~\ref{sec:Privacy-anal} uses this common behavior to mutualize the cost of countermeasures across clusters while amplifying privacy guarantees. We then empirically show in Section~\ref{sec:eval} that for typical execution plans the guarantees provided at the cluster level transfer well to the execution plan level.

\subsection{Performance Metrics}
\label{sec:metrics}
In the general case it is theoretically impossible to have both perfect privacy and low overhead (in terms of latency and bandwidth) and solutions have to trade one for the other, as shown in \cite{Trilemma}. While we study a more specific problem, solutions similarly lie on a spectrum from no privacy and no overhead to perfect privacy and huge overhead (e.g., broadcast).  In this work, we explore trade-offs between privacy, utility and efficiency, defined as follows:
\begin{itemize}
\item \emph{Privacy}, measured by the parameters $\epsilon$ and $\delta$ of differential privacy (which bound the amount of 
leakage about individual data point).
\item \emph{Utility}, measured as the number of 
tuples effectively used in the final result.
\item \emph{Efficiency}, divided into three distinct dimensions:
  \begin{itemize}
  \item Network load: the amount of additional traffic generated by our solution w.r.t. non-private communications,
  \item 
        Individual load on users' nodes: the number of secure channels per node,
  \item Additional users' consents: the number of additional user's nodes involved in the execution plan.
  \end{itemize}
\end{itemize}

Note that depending on the specific implementation and computation considered,
the relative importance of these metrics may vary. For instance, efficiency aspects may be crucial to make the approach practical when using secure hardware. Our solutions will provide simple ways to adjust the trade-off between these metrics.
\section{Solution by Local Sampling and Flooding}
\label{sec:DP-data}

Following the reduction described in Section~\ref{sec:simplifying-problem}, we consider a single cluster of nodes composed of a set $\Source$ of $S$ sources nodes, each source seeking to send one message to a set $\Target$ of $T$ target nodes. In this context, a communication graph $\mathcal{G}$ can be represented as a set of $S$ messages $\mathcal{G}=\{(s_1,t_1),\dots, (s_S,t_S)\}$ where each $s_i\in\Source$ and $t_i\in\Target$. Our goal is to design a differentially private algorithm $\Algo$ which takes as input a (non-private) communication graph $\mathcal{G}$ and returns a perturbed graph $\Algo(\mathcal{G})$ which preserves the utility and efficiency of the distributed computation as much as possible.

In this section, we propose and analyze a baseline solution in which each message $(s,t)$ is randomized at the source node $s$, independently of others. In other words, the algorithm $\Algo$ will be of the form $$\Algo(\mathcal{G})=\{\Randomizer(s_1,t_1),\dots,\Randomizer(s_S,t_S)\},$$
where $\Randomizer$ is a \emph{local randomizer} applied independently to each message.
This setting corresponds to the so-called local model of differential privacy \cite{LDP}.

\subsection{Proposed Algorithm}

Our algorithm is based on two principles: \emph{sampling} and \emph{flooding}.
Sampling consists in randomizing the target of the message, following the idea of randomized response \cite{DP:randomized-response,kRR}. Note however that in contrast to typical use-cases in which randomized response is used to perturb the content of the message, here we use it \emph{perturb the destination of the message}. Sampling allows to trade utility for privacy without affecting efficiency. On the other hand, flooding consists in producing additional dummy messages that are indistinguishable from real messages from the attacker point of view, but can be discarded by target nodes at execution time so they do not affect the final result of the computation. Flooding alone cannot guarantee differential privacy (except in the extreme case when it becomes equivalent to a broadcast), but combined to sampling we will show that it allows to trade efficiency for improved privacy without affecting utility.

Formally, let $\Plie\in[0,1]$ be the sampling parameter and $d\in \{1,\dots,T-1\}$ the flooding parameter. Our local randomizer $\Randomizer_\Plie^d$, executed at each source node $s$, takes as input a message $(s,t)$ and returns a collection of $d+1$ messages to be sent by the source node. These messages, which constitute the output observable by an adversary, are generated as follows: 
first, with probability $1 - \Plie$, the source node keeps the true message to the true target $t$, otherwise it sends a dummy message to a target $t'$ chosen uniformly at random;\footnote{If it happens that $t'=t$, we obviously send to true message to maximize utility.} then, the source node creates $d$ additional dummy messages aimed at a set of $d$ targets chosen uniformly without replacement from $\Target$ (excluding the target of the first message).
From this local randomizer $\Randomizer_\Plie^d$, we define the global algorithm as $\Algo_{\Plie}^{d}(\mathcal{G})=\{\Randomizer_\Plie^d(s_1,t_1),\dots,\Randomizer_\Plie^d(s_S,t_S)\}$, which applies $\Randomizer_\Plie^d$ to each message in $\mathcal{G}$. For convenience, we denote the sampling-only variants by $\Randomizer_\Plie:=\Randomizer_\Plie^0$ and $\Algo_\Plie:=\Algo_\Plie^0$.

\subsection{Privacy Analysis}

We can show the following differential privacy guarantees for $\Algo_{\Plie}^{d}$. The proof can be found in Appendix~\ref{app:proof-priv-local}.

\begin{theorem}
\label{thm:privacy_local}
Let $\Plie\in[0,1]$ and $d\in \{1,\dots,T-2\}$.
The algorithm $\Algo_{\Plie}^{d}$ satisfies $\epsilon$-differential
privacy with 
$$\epsilon=\ln\left( \frac{(1-\Plie)T}{\Plie ( d + 1 ) }+1 \right).$$
For the case where $d=T-1$, $\Algo_{\Plie}^{d}$ is equivalent to a broadcast and thus $\epsilon=0$.
\end{theorem}

\begin{figure}[t]
\centering
\begin{subfigure}{.4\textwidth}
\centering
    \includegraphics[width=.8\textwidth,keepaspectratio]{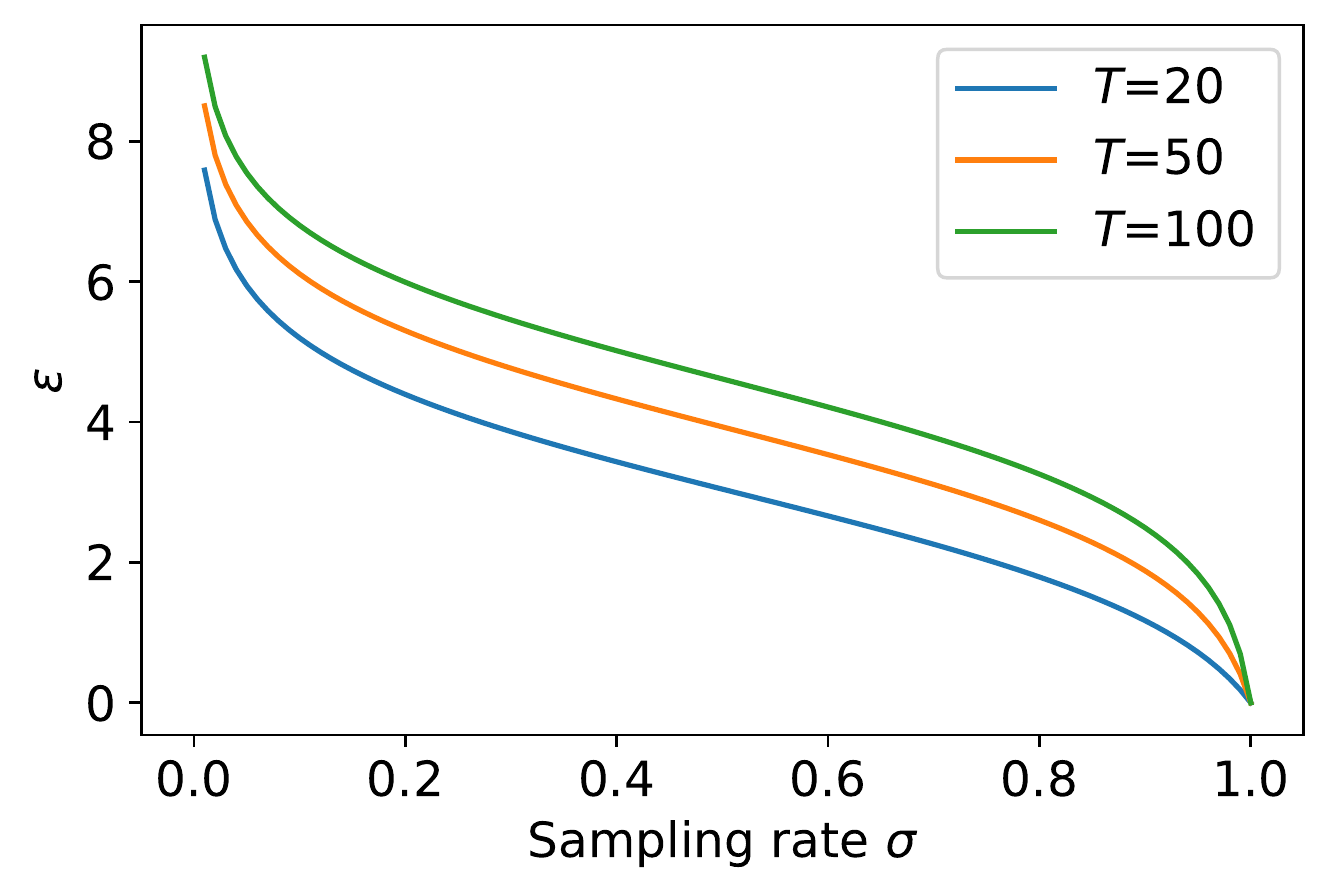}
    \caption{Privacy of $\Algo_{\Plie}$}
    \label{fig:Epsilon_Sampling}
\end{subfigure}
\hspace{.01\textwidth}
\begin{subfigure}{.4\textwidth}
\centering
    \includegraphics[width=.8\textwidth,keepaspectratio]{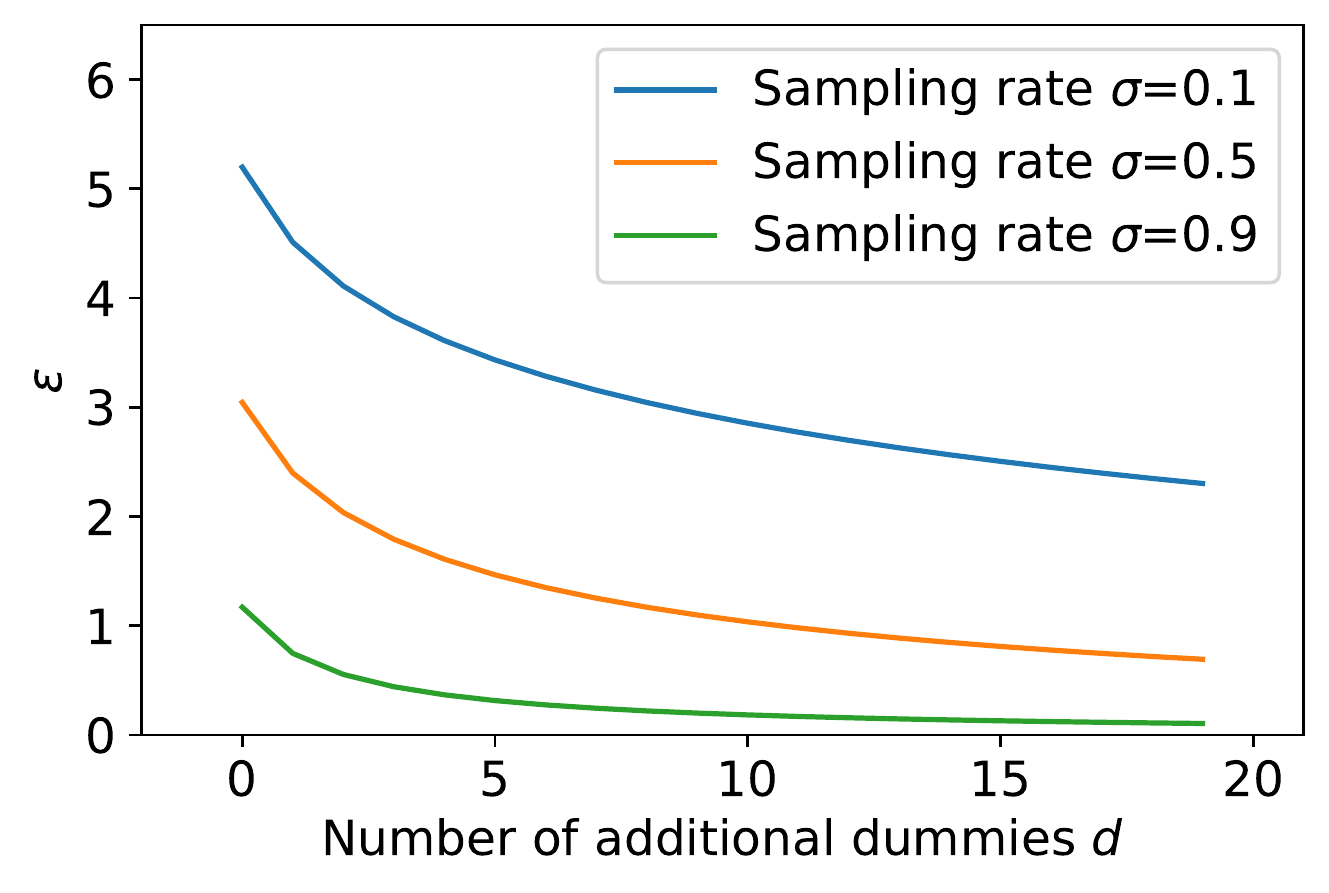}
    \caption{Privacy of $\Algo_{\Plie}^d$ ($T=20$)}
    \label{fig:with_dummies}
\end{subfigure}
\caption{Privacy of $\Algo_{\Plie}$ and $\Algo_{\Plie}^d$ varying $\Plie$ and $d$.}
\label{fig:dp1}
\vspace{-1em}
\end{figure}

To illustrate the influence of all parameters, we plot the values of $\epsilon$ as given in Theorem~\ref{thm:privacy_local} for different number of targets $T$, sampling parameter $\Plie$ and flooding parameter $d$.
Figure~\ref{fig:Epsilon_Sampling} shows the case of the sampling-only variant $\Algo_{\Plie}$ (no flooding). We see that achieving reasonable values of $\epsilon$ (typically, $\epsilon$ is recommended to be smaller than $\ln(3)$~\cite{dwork2008differential}) requires a large $\Plie$ (typically larger than $\Plie=0.9$), making $\Algo_{\Plie}$ quite impractical for use-cases that require low sampling (e.g., when recruiting additional consenting users to act as source nodes is very costly).
Flooding helps to reduce $\epsilon$ while keeping $\Plie$ fixed, as can be seen in Figure~\ref{fig:with_dummies}. 
Interestingly, there are diminishing returns: the bulk of the gains in privacy come from the first dummies. We also see that the bigger $\Plie$, the faster the decrease of $\epsilon$ with $d$.

\subsection{Performance Analysis}
\label{sec:perf_local}

We analyze the performance in terms of \emph{utility} and \emph{efficiency} metrics defined in Section~\ref{sec:metrics}.

\paragraph{Utility} In order to maintain the same utility (i.e., number of real contributions taken into account in the compute nodes) as the non-private algorithm with $S$ sources, the total number of source nodes must be $S_t=\frac{S}{(1-\Plie)}$, with also an impact on efficiency with $S_t-S$ additional users' consents.


\paragraph{Efficiency}
In $\Algo_{\Plie}^{d}$, $d+1$ messages are sent by each source to at most $d+1$ targets. The individual load  on computation (target) nodes in total number of secure communication channels to be initiated is hence bounded by $S_t \times (d+1)$ and the total volume of exchanged messages is $S_t \times (d+1) \times \mu$, with $\mu$ the size of a single message. Overall, $\Algo_{\Plie}^{d}$ introduces a factor $(d+1)$ overhead compared to a non-private execution. 

While effective to enforce high privacy while maintaining high utility, this baseline solution based on sampling and flooding thus comes at a significant cost in efficiency.

\section{Amplifying Privacy via Scramblers} 
\label{sec:Privacy-anal}

In the baseline algorithm proposed in the previous section, the privacy of each communication only comes from the sampling rate and flooding applied locally by the source node. In this section, we propose an approach where source nodes can benefit from the sampling applied at other source nodes as well as from shared dummy messages, thereby ``hiding in the crowd''. To this end, we introduce an additional type of node called \emph{scrambler} whose role is to collect a set of $n$ locally randomized messages (from $n$ source nodes), add $d$ extra dummy messages and shuffle the output before transmitting the messages to the target nodes.

After introducing our approach in Section~\ref{sec:scrambler}, we first state a pure $\epsilon$-DP guarantee in Section~\ref{sec:puredp}. Arguing that this result is quite conservative, in Section~\ref{sec:amplification} we prove $(\epsilon,\delta)$-DP guarantees which capture the desired ``hiding in the crowd'' effect, by extending recent results on amplification by shuffling \cite{Balle2019}. Finally, we briefly discuss the efficiency of the proposed approach in Section~\ref{sec:scrambler_perf}.

\subsection{Proposed Algorithm}
\label{sec:scrambler}

\begin{figure}[t]
\centering
\begin{subfigure}{.22\textwidth}
    \includegraphics[width=.7\textwidth,keepaspectratio]{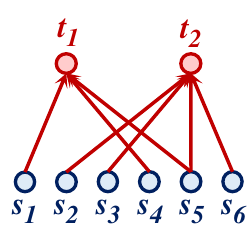}
    \caption{No scrambler (Section~\ref{sec:DP-data})}
    \label{fig:without_scramblers}
\end{subfigure}
\hspace{.01\textwidth}
\begin{subfigure}{.22\textwidth}
    \includegraphics[width=.7\textwidth,keepaspectratio]{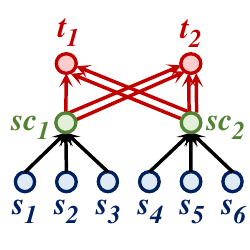}
    \caption{$2$ scramblers (Section~\ref{sec:Privacy-anal})}
    \label{fig:with_scramblers}
\end{subfigure}
\caption{Cluster with $2$ target, $6$ source and $2$ scrambler nodes.}
\label{fig:scramblers}
\vspace{-1em}
\end{figure}

One of the tools at our disposal to balance privacy, utility and efficiency is to add computation nodes. We propose to add a set of \emph{scrambler nodes}, which are assumed to have the same security properties as the other nodes and will be responsible for collecting messages from source nodes before transmitting them to target nodes. Recall that the content of messages is encrypted by source nodes and can only be decrypted by the target node. As relying on a single scrambler would introduce a single point of failure and require that the scrambler opens secure channels with each source node (which is impractical considering the load limitation constraints stated in Section~\ref{sec:problem}), we propose to add a set of $S/n$ \emph{scramblers} and assign each source node to one scrambler in an input-independent manner. The parameter $n$, assumed to divide $S$ for simplicity, thus corresponds to the number of source nodes assigned to each scrambler.
Figure~\ref{fig:scramblers} shows the difference between this new architecture and the one considered in Section~\ref{sec:DP-data}.

\begin{algorithm2e}[t]
  \DontPrintSemicolon
   \LinesNumbered
  \SetKwComment{Comment}{{//\ \scriptsize}}{}
  \newcommand{\comr}[1]{\Comment*[r]{#1}}
  \caption{Algorithm $\mathcal{S}_d$ executed at each scrambler}
  \label{alg:scrambler}
  \KwIn{Set of messages $\{(s_i,t'_i)\}_{i=1}^n$, flooding parameter $d$, list of potential targets $\Target$, boolean $B$ (optional: to cap the number of messages per target, default \texttt{False})}
  \KwOut{set of messages $\mathcal{O}$}
  \BlankLine
%
  $\Output = \{t'_i\}_{i=1}^n$\;
  \For(\texttt{~//\ add $d$ dummy messages}){$j\gets1$ \KwTo $d$}{
     \If{$B=$\texttt{True}}{
        $\Target \leftarrow \Target \setminus \{t : \Output.\texttt{count}(t) = n\}$ 
     }
     $t_j \leftarrow$ uniformly random target from $\Target$\;
     $\Output \leftarrow \Output \cup \{t_j\}$
  }
  $\Output$.\texttt{shuffle()} \Comment{random permutation}
  \Return{$\Output$}
\end{algorithm2e}

We now describe the algorithm we propose in this setting. Let $\Plie\in[0,1]$ be the sampling parameter and $d\in \mathbb{N}$ the flooding parameter. Given an input communication graph $\mathcal{G}=\{(s_1,t_1),\dots, (s_S,t_S)\}$, each message $(s_i,t_i)$ is first processed by its source node $s_i$ who applies the local randomizer $\Randomizer_\Plie(s_i,t_i)$ (i.e., with sampling rate $\Plie$ but no flooding) and sends the resulting message $(s_i,t_i')$ to the scrambler. Then, each scrambler collects the $n$ inputs from its assigned sources, adds $d$ dummy messages, shuffles the whole set of messages and sends them to the corresponding targets.
The scrambler algorithm is shown in Algorithm~\ref{alg:scrambler}. Note that the default behavior for creating dummy messages is to select targets uniformly at random with replacement, which we denote by $\mathcal{S}_d$. We will also consider a variant where the number of messages sent to each target is capped by $n$, which we denote by $\Sbar_d$. In any case, the scrambler can be implemented by a simple shuffling primitive, which is becoming standard in the design of private systems \cite{Prochlo,Cheu2019,amp_shuffling,Balle2019,esa2}.

Crucially, the communication between sources and scramblers is input-independent, therefore from the point of view of the adversary the output can be restricted to the communication between scramblers and targets.
Furthermore, since each scrambler operates over a distinct partition of $n$ source nodes, from the differential privacy point of view it will be sufficient to analyze the algorithm at the level of a single partition. \nicolas{dire explicitement que le résultat en trme de DP est le même pour un ensemble de scramblers donc pour le cluster de noeuds} We denote the partition-level algorithm by $\Algo_\Plie^{n,d}$, which can be written as a sequential composition of the local randomizer $\Randomizer_\Plie$ and the scrambler algorithm $\mathcal{S}_d$:
\begin{equation}
\label{eq:shuffle_alg_decompo}
\Algo_\Plie^{n,d}(\mathcal{G}) = \mathcal{S}_d\big(\Randomizer_\Plie(s_1,t_1),\dots,\Randomizer_\Plie(s_n,t_n)\big).
\end{equation}
Similarly, we denote by $\Abar_\Plie^{n,d}$ the variant based on $\Sbar_d$.

Intuitively, $\Algo_\Plie^{n,d}$ and $\Abar_\Plie^{n,d}$ should achieve a better privacy-utility-efficiency trade-off than the baseline approach of Section~\ref{sec:DP-data} because an adversary can only infer information from the communication pattern between the scrambler and the targets. In particular, (i) each dummy message added by the scrambler should improve privacy for all sources nodes in the partition, and (ii) local sampling at each source node should further help to hide each message destination. This is what we will examine in our privacy analysis.

\subsection{Privacy Analysis: Pure $\epsilon$-DP}
\label{sec:puredp}


Our first result quantifies the privacy guarantees provided by the algorithm in terms of (pure) $\epsilon$-DP. For dummies to provide some benefit in terms of $\epsilon$-DP, we need to consider the variant $\Abar_\Plie^{n,d}$ where the number of messages to each target is capped by $n$. We have the following result.

\begin{theorem}
Let $\Plie\in[0,1]$ and $d\in \{1,\dots,n-1\}$.
The algorithm $\Abar_\Plie^{n,d}$ satisfies $\epsilon$-differential
privacy with 
$$e^\epsilon=\frac{  \sum_{k=0}^{d}{\binom{d}{k} \binom{n-1}{k} (1-\Plie)^{k}{R_{lie}}^{n-k-1} \Big( 1-\Plie+ k \cdot  \frac{{R_{lie}}^{2}}{1-\Plie} \Big)  }}{\sum_{k=0}^{d}{\binom{d}{k} \binom{n-1}{k} (1-\Plie)^{k}{R_{lie}}^{n-k-1} \Big( R_{lie}+ k \cdot \frac{{R_{lie}}^{2}}{1-\Plie}  \Big) }} $$
where $R_{lie}=\tinylie$ and $\binom{n}{k}$ is the binomial coefficient.
\label{theorem-exact}
\end{theorem}

\begin{proof}[Sketch of proof]
The main challenges are to deal with the combinatorial number of possible inputs and outputs, and the fact that each dummy message depends on the input messages as well as previously drawn dummies (due to the cap on the maximum number of messages per target). Our proof is based on factorizing the probability distribution of outputs in a way that allows us to identify the combination of neighboring communication graphs $\mathcal{G},\mathcal{G}'$ and output $\mathcal{O}$ which produces the worst-case ratio $P[\Algo_\Plie^{n,d}(\mathcal{G})=\Output]/P[\Algo_\Plie^{n,d}(\mathcal{G}') = \Output]$. We can then compute the exact value of this ratio, which gives $\epsilon$.
\inPreprint{The detailed proof can be found in Appendix~\ref{app:Appendix1}.}{}
\end{proof}

\begin{figure*}[!t]
\centering
\begin{subfigure}{.32\textwidth}
    \includegraphics[width=\textwidth,keepaspectratio]{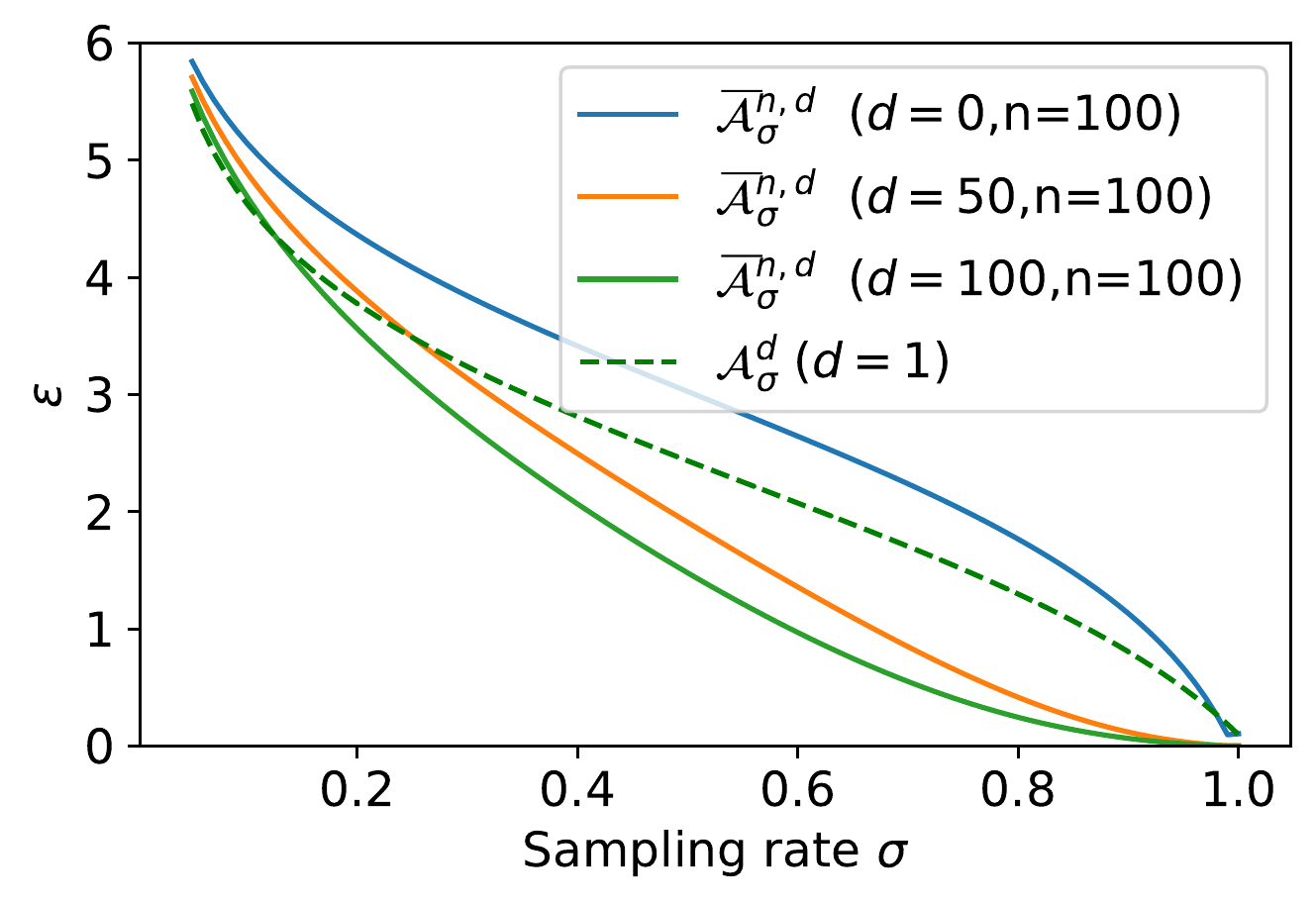}
    \caption{Varying sampling rate $\Plie$ ($n=100$)}
   \label{fig:nbound-plie}
\end{subfigure}
\begin{subfigure}{.32\textwidth}
    \includegraphics[width=\textwidth,keepaspectratio]{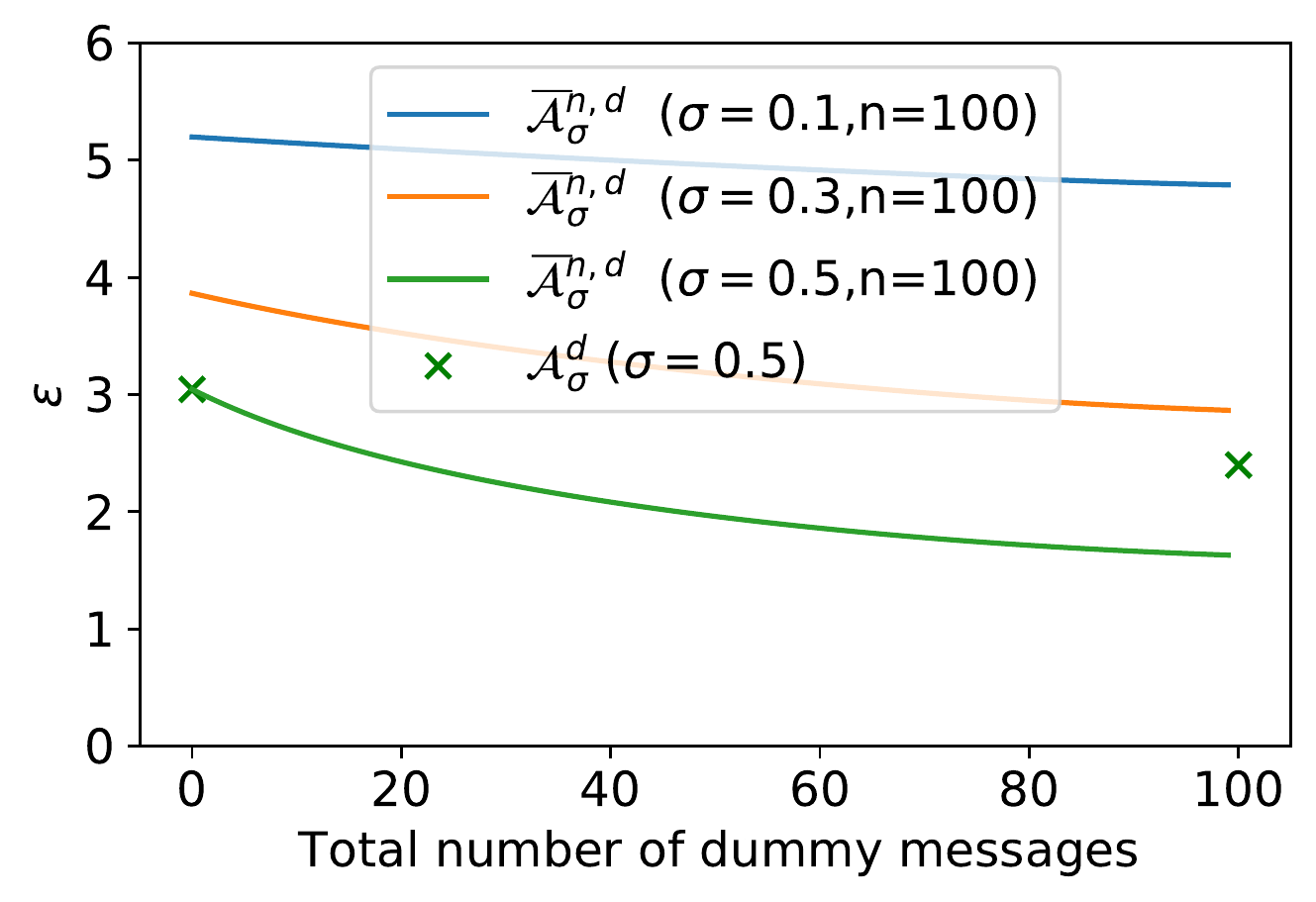}
    \caption{Varying number $d$ of dummies ($n=100$)}\label{fig:nbound-d}
\end{subfigure}
\begin{subfigure}{.32\textwidth}
    \includegraphics[width=\textwidth,keepaspectratio]{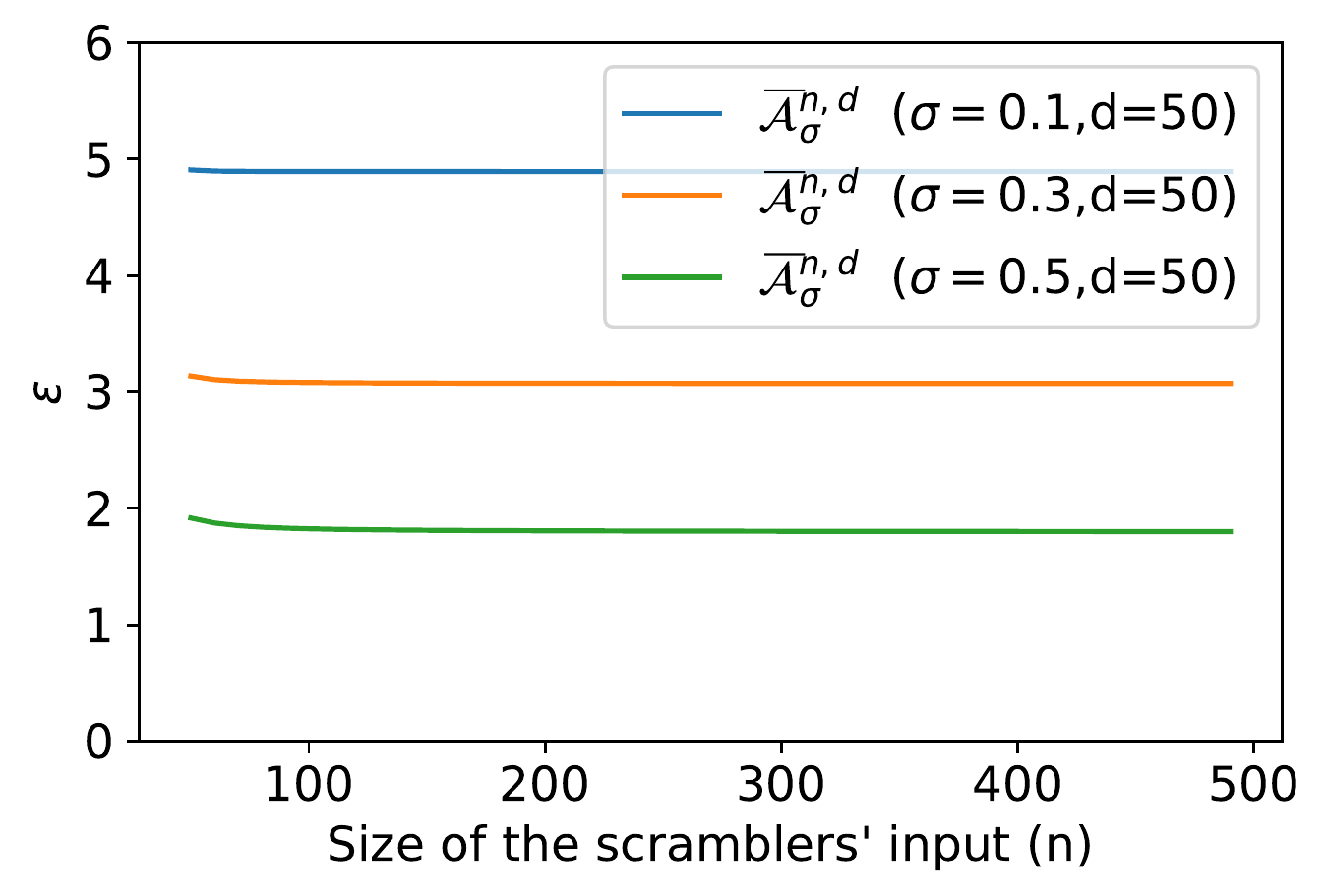}
    \caption{Varying number $n$ of messages ($d=50$)}
    \label{fig:nbound-n}
\end{subfigure}
\caption{Impact of the parameters $\sigma$, $d$ and $n$ on the privacy of $\protect\Abar_\Plie^{n,d}$, measured by $\epsilon$, for $T=20$ targets. For comparison purposes, we also plot the privacy of $\Algo_{\Plie}^{d}$ (the local algorithm without scrambler), which adds $nd$ dummy messages in total.}
\label{fig:nbound}
\vspace{-1em}
\end{figure*}
    
As the formula in Theorem~\ref{theorem-exact} is difficult to interpret, we plot the value of $\epsilon$ when varying the parameters $\Plie$, $d$ and $n$ in Figure~\ref{fig:nbound}.
\Cref{fig:nbound-d,fig:nbound-plie} confirm that the privacy guarantees provided by $\Abar_\Plie^{n,d}$ increase with the local sampling rate $\Plie$ and the number of dummies $d$ added by the scrambler. More interestingly, we see that $\Abar_\Plie^{n,d}$ provides stronger privacy than the local algorithm $\Algo_{\Plie}^{d}$ (without scrambler) when compared at equal sampling rate $\sigma$ and total number of dummies. Equivalently, $\Abar_\Plie^{n,d}$ can match the privacy of $\Algo_{\Plie}^{d}$ with fewer dummy messages, i.e., with better efficiency.

\inPreprint{
\begin{figure}[!t]
\centering
\includegraphics[width=.4\textwidth,keepaspectratio]{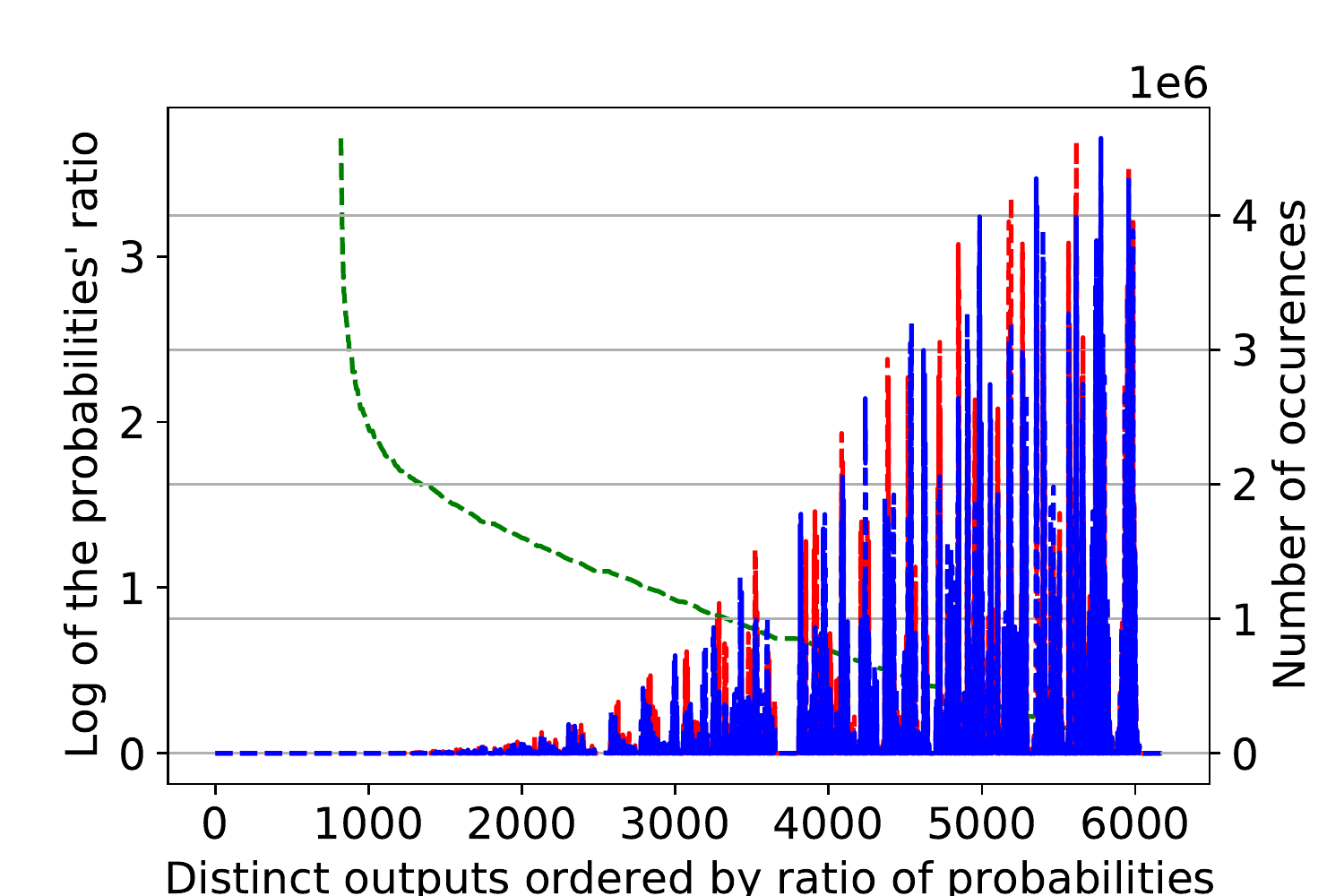}
\caption{Number of occurrences and corresponding log-ratio of probabilities (in decreasing order) for the distinct outputs obtained across $7\times 10^8$ runs of $\Algo_\Plie^{n,d}$ ($T=4$, $\Plie=0.2$, $d=20$ and $n=20$). See main text and Appendix~\ref{app:simu} for details.}
\label{fig:exp}
\end{figure}
}{}

However, \Cref{fig:nbound-n} shows that the number of messages $n$ given as input to the scrambler does not have a significant impact on the privacy guarantee. This is due to the fact that pure $\epsilon$-DP is governed by the ratio of probabilities of the worst-case output, even if the probability of that output actually occurring is extremely small.
\inPreprint{}{These extremely improbable bad events can be accounted for by the parameter $\delta$ of $(\epsilon, \delta)$-differential privacy.}
\inPreprint{
To illustrate the fact that $\epsilon$-DP is quite restrictive and gives a somewhat pessimistic view of the privacy guarantees provided by our algorithms, we performed a numerical simulation on a small problem instance for a large number of random runs (see Appendix~\ref{app:simu} for details on this simulation).
Figure~\ref{fig:exp} shows how many times each output occurred across runs of algorithm $\Algo_\Plie^{n,d}$ on two neighboring communication graphs (the red and blue histograms) and gives the estimated log-ratio of probabilities for each of these outputs (the green dots). The results clearly show the largest log-ratios (i.e., the ones that drive up the value of $\epsilon$) correspond to very rare outputs. This was observed consistently across all inputs we tried.

Our simulation results suggest that Theorem~\ref{theorem-exact} may not fully reflect the privacy benefits of our algorithms. 
}{}
In the next section, we prove $(\epsilon,\delta)$-differential privacy results which give a more faithful account of the privacy guarantees of $\Algo_\Plie^{n,d}$.


\subsection{Privacy Analysis: $(\epsilon,\delta)$-DP via Amplification by Shuffling}
\label{sec:amplification}


The goal of this section is to prove $(\epsilon,\delta)$-DP guarantees for our algorithm $\Algo_\Plie^{n,d}$. Obtaining such guarantees is more challenging than proving $\epsilon$-DP, as the latter boils down to identifying and computing (or upper bounding) the worst-case ratio of probabilities over all possible outputs. In contrast, establishing $(\epsilon,\delta)$-DP requires to characterize the distribution of outputs so that its ``long tails'' 
(where the desired $\epsilon$-DP guarantees may not hold) can be accounted for in $\delta$.

Recall that $\Algo_\Plie^{n,d}$ can be seen as the composition of two algorithms: a local randomizer $\Randomizer_\Plie$ and a scrambler $\mathcal{S}_d$ (see Eq.~\ref{eq:shuffle_alg_decompo}).
To prove $(\epsilon,\delta)$-DP guarantees, we leverage and extend techniques from the recent literature on privacy amplification by shuffling \cite{Cheu2019,amp_shuffling,Balle2019}. Privacy amplification by shuffling allows to prove differential privacy guarantees for algorithms of the form $\Algo=\mathcal{S}_0\circ\Randomizer^n$, where $\Randomizer$ is any differentially private local randomizer and $\mathcal{S}_0$ simply returns a shuffle of its inputs. Privacy amplification by shuffling was designed to provide privacy for the \emph{content} of messages. Here, we use this idea in the novel context of providing privacy for the \emph{communication patterns}, and also extend the proof of \cite{Balle2019} to account for the additional privacy brought by dummy messages added by our scrambler $\mathcal{S}_d$.


A key quantity in our analysis is the so-called \emph{privacy amplification random variable} introduced by \cite{Balle2019}. In our context, given $\epsilon>0$, two messages $(s,t)$ and $(s,t')$ with $t\neq t'$ and a random variable $t''\sim\text{Uniform}(\{1,\dots,T\})$ uniformly distributed over the set of targets, the privacy amplification random variable can be written as:
\begin{align}
    L_\epsilon^{t,t'} 
    &= \sigma(1-e^\epsilon) + (1-\sigma)T(\I[t''=t] - e^\epsilon \I[t''=t']),\label{eq:privacy_amp_var_main}
\end{align}
where $\I$ is the indicator function.
$L_\epsilon^{t,t'}$ is related to the difference in the probabilities that $\Randomizer_\Plie$ outputs $(s,t'')$ when applied to two different inputs $(s,t)$ and $(s,t')$. It will quantify the contribution of each input-independent random message present in the final output (be it obtained from sampling by source nodes or from dummy messages added by the scrambler) to the $(\epsilon,\delta)$-differential privacy guarantees of $\Algo_\Plie^{n,d}$.

\paragraph{Analytical bound} We are now ready to state the main result of this section: an analytical $(\epsilon,\delta)$-DP guarantee for $\Algo_\Plie^{n,d}$ (the proof can be found in Appendix~\ref{app:shuffling_proof}).
\begin{theorem}
\label{thm:main_amplification}
Let $d\in \mathbb{N}$. If $\Plie>0$, then the algorithm $\Algo_\Plie^{n,d}$ satisfies $(\epsilon,\delta)$-differential privacy for any $\epsilon>0$ with
$$\delta \geq \frac{1}{\Plie n}\sum_{m=1}^{n} \frac{m}{m+d}\binom{n}{m}\Plie^{m}(1-\Plie)^{n-m} \frac{b^2}{4a}e^{-\frac{2(m+d)a^2}{b^2}},$$
where $a=1-e^\epsilon$ and $b=(1-\Plie)T(1+e^\epsilon)-2\Plie(1-e^\epsilon)$.

If $\Plie=0$, then $\Algo_\Plie^{n,d}$ satisfies $(\epsilon,\delta)$-DP for any $\epsilon>0$ with
$$\delta \geq \frac{1}{d+1} \frac{b^2}{4a}e^{-\frac{2(d+1)a^2}{b^2}}.$$
\end{theorem}

\begin{remark}[Generality of our analysis]
\label{rem:general}
Beyond the particular case of the local randomizer $\Randomizer_\Plie$ we use in this work, our analysis readily applies to any other differentially private local randomizer $\Randomizer$. The only condition is that dummy messages drawn by $\mathcal{S}_d$ must follow a specific distribution which depends on the choice of $\Randomizer$, see Appendix~\ref{app:shuffling_proof} for details.
\end{remark}

Theorem~\ref{thm:main_amplification} highlights the trade-off between $\epsilon$ and $\delta$. Given a value for $\epsilon$, the formulas directly give the lowest admissible value of $\delta$. Conversely, we can fix $\delta$ and numerically search for the lowest value of $\epsilon$ for which the inequality is satisfied. As $\delta$ can be interpreted as the probability that the privacy loss exceeds $\epsilon$, it should be kept to a small value.
A general guideline for $\delta$ is that it must be smaller than $1/n$ to provide meaningful guarantees \cite{Dwork2014a}.
Crucially, Theorem~\ref{thm:main_amplification} shows that the privacy guarantees improve with the number $n$ of messages, as each original message will become input-independent with probability $\sigma$. Dummy messages further contribute by guaranteeing a minimal number $d$ of input-independent messages in the output. As shown by the second inequality, we can obtain $(\epsilon,\delta)$-DP guarantees even if $\sigma=0$, i.e., without affecting the utility but only the efficiency.


For the sake of simplicity, Theorem~\ref{thm:main_amplification} relies on Hoeffding's inequality to upper bound the sum of privacy amplification random variables. Numerically tighter (albeit more complex) bounds can be obtained by applying Bennett's inequality (which leverages the variance of $L_\epsilon^{(t,t'}$), see Appendix~\ref{app:shuffling_proof} for details.

\begin{figure*}[!t]
\centering
\begin{subfigure}{.32\textwidth}
    \includegraphics[width=\textwidth,keepaspectratio]{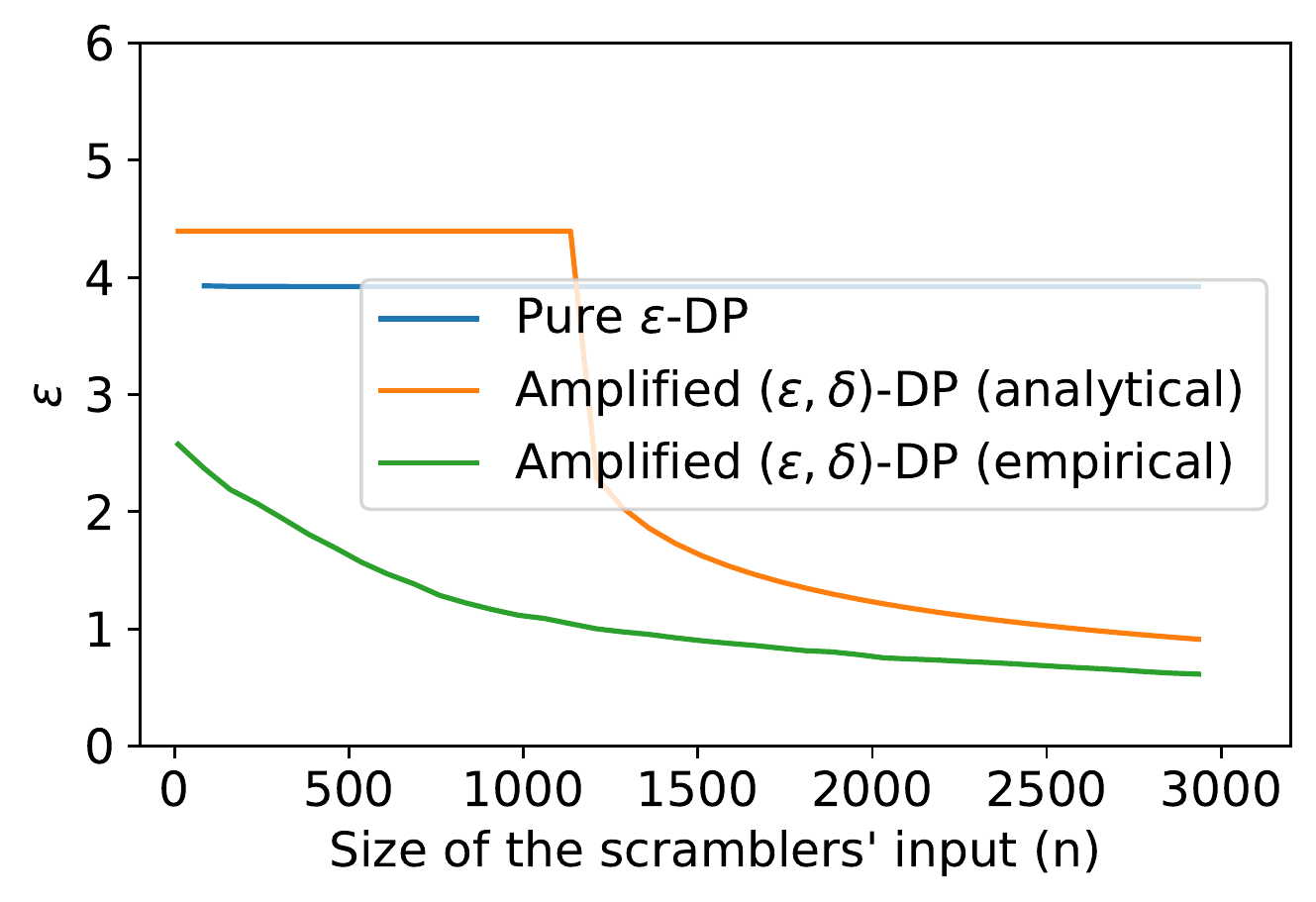}
    \caption{Impact of $n$ ($d=50$, $\sigma=0.2$)}\label{fig:innboundComparison-n}
\end{subfigure}
\begin{subfigure}{.32\textwidth}
    \includegraphics[width=\textwidth,keepaspectratio]{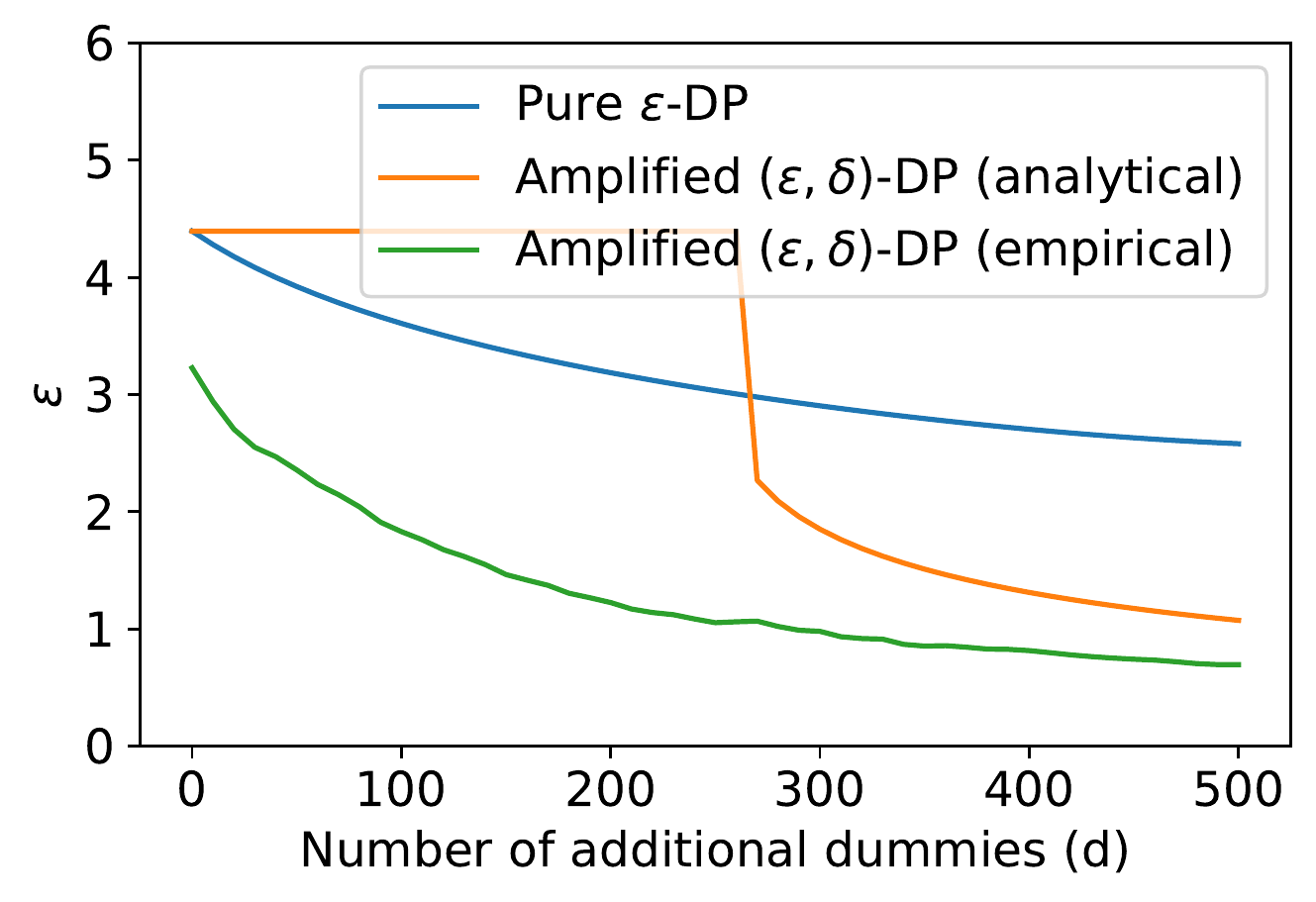}
    \caption{Impact of $d$ ($n=100$, $\sigma=0.2$)}\label{fig:innboundComparison-d}
\end{subfigure}
\begin{subfigure}{.32\textwidth}
    \includegraphics[width=\textwidth,keepaspectratio]{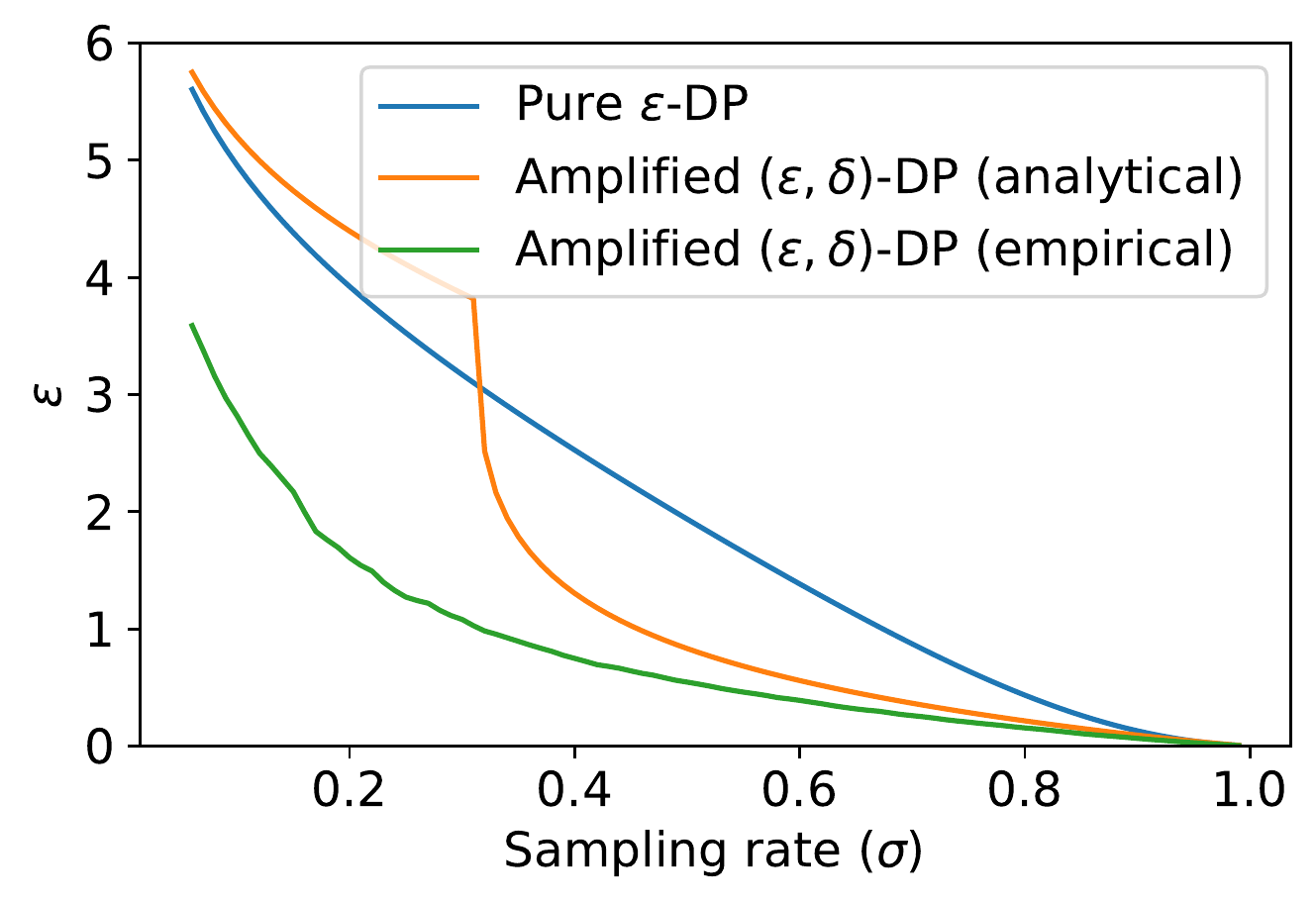}
    \caption{Impact of $\Plie$ ($n=500$, $d=50$)}\label{fig:innboundComparison-plie}
\end{subfigure}
\caption{Impact of $n$, $d$ and $\Plie$ on the privacy of $\Algo_\Plie^{n,d}$, measured by $\epsilon$, for $T=20$ targets. For $(\epsilon,\delta)$-DP, we set $\delta=10^{-4}$.}
\gs{mettre les figures à la même échelle \riad{Done, c'est bon comme ça ?}}
\label{fig:innboundComparison}
\vspace{-1em}
\end{figure*}

We use this improved version of our analytical bound to plot the privacy guarantees of $\Algo_\Plie^{n,d}$ as a function of the different parameters and comparing to our previous $\epsilon$-DP result (Theorem~\ref{theorem-exact}).
Figure~\ref{fig:innboundComparison} shows that our amplified $(\epsilon,\delta)$-DP guarantees provide large improvements over the previous $\epsilon$-DP result in regimes where the expected number of input-independent messages in the output (roughly $\sigma n+d$) is sufficiently large (in the order of $300$). When $\sigma n$ and $d$ are both small, Theorem~\ref{theorem-exact} does not provide any privacy improvement and we fall back on the privacy guarantees provided by the local randomizer $\Randomizer_\Plie$ alone.

\paragraph{Further improvements} The result of Theorem~\ref{thm:main_amplification} is in fact quite pessimistic when both $\sigma n$ and $d$ are small: this is because the concentration inequalities used to bound $\mathbb{E}[\sum_{i=1}^{m+d+1} L_i ]_+$ are known to be loose when the number $m+d+1$ of terms in the sum is small.
To get tighter privacy guarantees in such regimes, we can instead compute a Monte Carlo estimate of $\mathbb{E}[ \sum_{i=1}^{m+d+1} L_i ]_+$, i.e., we can approximate the expectation empirically using a finite number $R$ of random samples. The procedure is outlined in Algorithm~\ref{alg:empirical_est}. Note that drawing a sample of the privacy amplification random variable amounts to fixing two arbitrary targets $t\neq t'$, sampling a target $t''$ uniformly at random from $\Target$, and computing Eq.~\ref{eq:privacy_amp_var_main}.

The larger $R$, the closer the empirical estimate $\hat{L}$ is to the true value. We can bound the deviation $|\hat{L} - \mathbb{E}[\sum_{i=1}^{m+d+1} L_i ]_+|$ with high probability using concentration inequalities, which gives a high probability bound on the error in estimating $\delta$. In all our plots and experiments we use $R=5000$, which is sufficient for Hoeffding's inequality to ensure that the probability of the relative estimation error of $\delta$ being larger than $1/1000$ is negligibly small (below $10^{-30}$).

\begin{algorithm2e}[t]
  \DontPrintSemicolon
   \LinesNumbered
  \SetKwComment{Comment}{{//\ \scriptsize}}{}
  \newcommand{\comr}[1]{\Comment*[r]{#1}}
  \caption{Empirical estimation of $\mathbb{E}[ \sum_{i=1}^{m+d+1} L_i ]_+$}
  \label{alg:empirical_est}
  \KwIn{Number of random samples $R$}
  \BlankLine
  \For{$r\gets1$ \KwTo $R$}{
     Draw $L^r_1,\dots,L^r_{m+d}$ according to Eq.~\ref{eq:privacy_amp_var_main}\;
     $L^r \leftarrow \big[ \sum_{i=1}^{m+d} L^r_i \big]_+$
  }
  $\hat{L}\leftarrow \frac{1}{R}\sum_{r=1}^R L^r$\;
  \Return{$\hat{L}$}
\end{algorithm2e}

Figure~\ref{fig:innboundComparison} shows the strong gains in the privacy guarantees obtained using this empirical estimation: we are able to obtain significant privacy amplification compared to pure $\epsilon$-DP (Theorem~\ref{theorem-exact}) even in regimes where both $\sigma n$ and $d$ are small.

In summary, we have derived analytical $(\epsilon,\delta)$-DP guarantees for our algorithm $\Algo_\Plie^{n,d}$ by leveraging and extending techniques from the literature of privacy amplification by shuffling. We have also shown how to obtain tighter empirical bounds. We will see in Section~\ref{sec:eval} how to use our results to tackle practical use-cases.

\subsection{Performance Analysis}
\label{sec:scrambler_perf}

As before, recall that we consider a simple computation with $S=|\Source|$ source nodes delivering a single message to one of the $T=|\Target|$ potential target nodes.

\paragraph{Utility} In order to maintain the same utility (i.e., number of real contributions) as the non-private algorithm with $S$ sources, the total number of source nodes must be $S_t=\frac{S}{(1-\Plie)}$. We thus consider $SF=S_t/n$ scramblers.

\paragraph{Efficiency} Each scrambler must open a secure communication channel with $n$ sources and $T$ targets and exchanges $2\times n+d$ messages. Hence the total number of secure channels is $S_t+SF\times T$ and the volume of exchanged messages is $S_t + SF \times(n+d)\times \mu$.

\section{Evaluation} 
\label{sec:eval}



User-side collaborative computing is gaining interest with the emergence of (1) cross-device federated learning~\cite{kairouz2019advances,fed_systems} where large sets of personal devices collaboratively train machine learning models, and (2) personal database management systems~\cite{urquhart2018realising, DBLP:journals/is/AnciauxBBNPPS19} where populations of trusted user devices are engaged in collective database aggregation queries~\cite{ladjel2019trustworthy,loudet2019sep2p}. 


Data-dependent communication schemes increase performance by distributing data to compute nodes based on the data values of group-by/join keys (e.g., parallel Oracle SQL Analytics~\cite{bellamkonda2013adaptive}), data points distance to given centroids or regions of feature space (e.g., parallel $K$-means~\cite{zhao2009parallel}, $K$-medoids~\cite{kmedoids}, DBSCAN~\cite{RP-DBSCAN}) or similarity of users' profiles~\cite{clustfl1,clustfl2}.



We focus on two types of distributed queries representative of these contexts, and illustrate the trade-offs between privacy, utility and efficiency obtained with our proposal.

\subsection{Queries and Datasets} 
\label{sec:data-and-queries}


We consider two queries representative of above cases, called \emph{Aggregate} and \emph{$K$-means}. Aggregate is used to understand frequency distributions and collect marginal statistics from a set of consenting users (see example in Section~\ref{sec:usecase}). $K$-means is representative of iterative data processing algorithms used for instance in data mining and machine learning. We describe the two queries and their execution plans below.


\emph{Aggregate (see Section~\ref{sec:usecase}).} A set of $C \times G$ compute nodes, with $G$ the number of grouping sets of attributes (e.g., $G=2$ in Example~\ref{ex:use-case-full}), evaluates statistical functions (min, max, avg), with each compute node processing a partition of the input dataset and sending its output to the result node. A set of $S$ source nodes, each holding a single tuple with numeric values (on which statistics are computed) and grouping values (according to which tuples are grouped), send values from their tuple to $G$ compute nodes. 

\emph{$K$-means:} A set of $S$ source nodes, each holding a single data tuple, compute the distance of their tuple to the $K$ centroids and send their tuple to the compute node managing the closest centroid. A set of $C \times I$ compute nodes (with $I$ the number of iterations, and $C=K$ the number of centroids), each managing a single centroid for a single iteration, update their centroid using the data tuples received from the source nodes, send their updates to the compute nodes for the next iteration, and propagate the updated centroids back to the source nodes they interact with. These steps are repeated for a fixed number $I$ of iterations. The final result is transmitted to the result node. The initial state (first iteration) consists of $K$ (random) points representing the initial centroids of $K$ clusters.

In both cases, communication patterns between source and compute nodes reveal potentially sensitive information about the input data (grouping keys in Aggregate and close/similar users' tuples in $K$-means). Our proposal adds local sampling at source nodes and a set of $\mathrm{SC}$ scrambler nodes per grouping set/iteration between the source and compute nodes (see Figure~\ref{fig:with_scramblers}). For $K$-means, the new centroid obtained by each compute node at the end of the current iteration is sent back to all scrambler nodes which propagate them back to the source nodes they interact with to initiate the next iteration.

\emph{Datasets.} For Aggregate, we use a synthetic dataset. We tested both uniform and biased distributions: for both cases we generated from 10k to 100k records distributed across $20$ grouping intervals for 4 grouping attributes. For $K$-means, we use the classic MNIST dataset of handwritten digits, composed of 70k records in dimension 784 and distributed among 10 classes. We execute $K$-means on the training set (60k), and measure the quality of the clusters we obtain on the test set (10k) using the rand index metric to compare to the ground-truth class labels and evaluate utility.

\subsection{Practical Trade-offs and Results}

\emph{Privacy evaluation.} 
In both execution plans any source node potentially sends messages (via scrambler nodes) to any compute node, and any partition of mutually disjoints sets of source nodes can define a set of clusters of nodes. Each of the $\mathrm{SC}$ 
scrambler nodes associated to a given grouping set (in Aggregate) or iteration (in $K$-means) takes as 
input a partition of the source nodes and hence belongs to different clusters that are never on the same data path. On the contrary, 
scrambler nodes assigned to successive iterations or different grouping sets use as inputs the same sets 
of input source nodes and are therefore on the same data path. In Aggregate, according to Theorem~\ref{th:compo-clusters} the privacy of the 
overall execution plan is thus given by $\epsilon \leq G\times \underset{1\leq i\leq\mathrm{SC}}{\max}(\epsilon_{i})$ and 
$\delta \leq G\times \underset{1\leq i\leq\mathrm{SC}}{\max}(\delta_{i})$ where $\epsilon_i$ and $\delta_i$ denote the privacy guarantees for the cluster with the $i$-th scrambler.
Similarly, the privacy of the execution plan for $K$-means is $\epsilon=I\times 
\underset{1\leq i\leq\mathrm{SC}}{\max}(\epsilon_{i})$ and $\delta=I\times \underset{1\leq i\leq\mathrm{SC}}{\max}(\delta_{i})$.

\emph{Parameters.} Different trade-offs between privacy, utility, and performance can be
studied. The parameters of the experiments are shown in Table~\ref{tab:experiments-param}. 
$C$ is the number of compute nodes and determines the degree of parallelism of the algorithm. Its value depends on the computation and cannot be changed without 
impacting the performance. The value of the number $S$ of source nodes (consenting users) 
influences both efficiency and utility. It is varying in our experiments between 10k and 
20k. $\mathrm{SC}$ is the number of scrambler nodes involved per grouping set (for Aggregate, the 
number $G$ of grouping sets varies from 1 to 4) and per iteration (for $K$-means, the number 
$I$ of iterations is $10$). Parameters $\epsilon$ and $\delta$ determine privacy. The numbers $n$ of source nodes per scrambler, and $d$ of dummy messages added per scrambler, affect 
privacy and efficiency, while the sampling rate $\sigma$ affects privacy and utility. Fixing 
two of these three parameters and increasing the third would increase privacy.  

\emph{Results.} 
To study the impact of the different parameters in the trade-offs between privacy, utility and efficiency metrics (see Section~\ref{sec:metrics}), for simplicity, we work at fixed utility, i.e., we fix the number of tuples effectively used by a compute node (except in Fig.~\ref{fig:gb-utility} which varies utility). We then plot privacy as a function of the efficiency metrics along its three dimensions: 
(i) the \emph{network load} overhead, evaluated as the number of messages added compared to a regular distributed execution, depends on the number of scrambler nodes $\mathrm{SC}$ added per grouping set or iteration, number of dummies $d$ introduced by each scrambler node, and number of source nodes added for privacy reasons,  
(ii) the \emph{individual load}, evaluated as the number of secure channels created per node, which depends mainly on the number $n$ of source nodes per scrambler node\footnote{And to a lesser extent on the number $C$ of compute nodes, as $C$ is small.} and is limited by the power/bandwidth of end-user devices,
(iii) the number of additional \emph{users' consents} required, influenced by the sampling rate $\sigma$ and the number of contributors $S$. 
Note that we consider the three metrics for Aggregate, but only the first two for $K$-means. Indeed, the impact of sampling on utility is known to be negligible for $K$-means \cite{sculley2010web}. Preliminary measures using the rand index metric against the ground-truth clusters confirm that the proportion of correctly clustered tuples for high sampling rates ($\sigma=0.9$) is very close to that obtained when we do not sample. 

\begin{figure*}[!t]
\centering
\begin{subfigure}{.32\textwidth}
    \includegraphics[clip, trim=0 0 0 1cm, 
                width=\textwidth,keepaspectratio]{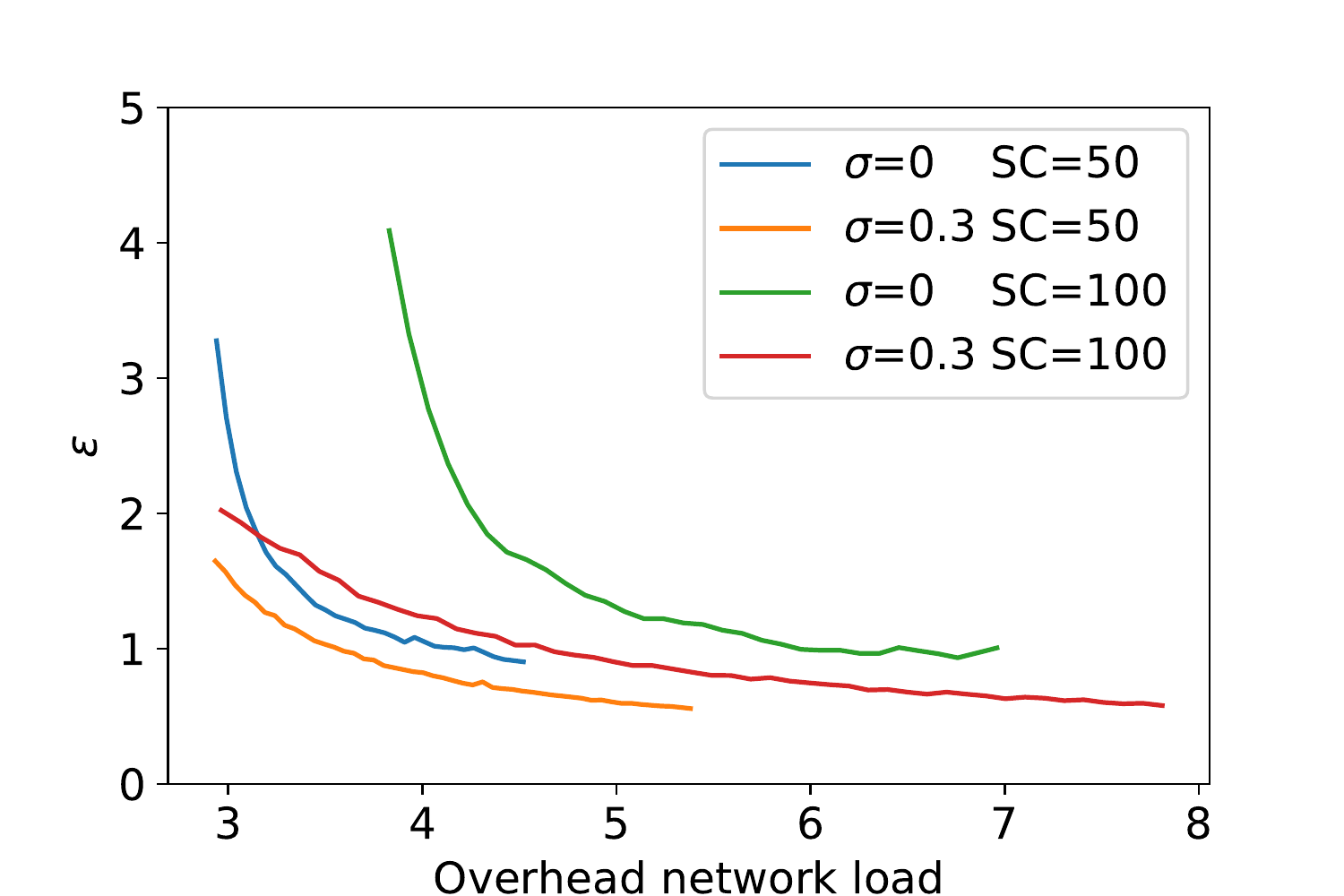}
    \caption{Aggregate: privacy vs network load ($G$=1)}
   \label{fig:gb-global-load}
\end{subfigure}
\begin{subfigure}{.32\textwidth}
    \includegraphics[clip, trim=0 0 0 1cm, 
                width=\textwidth,keepaspectratio]{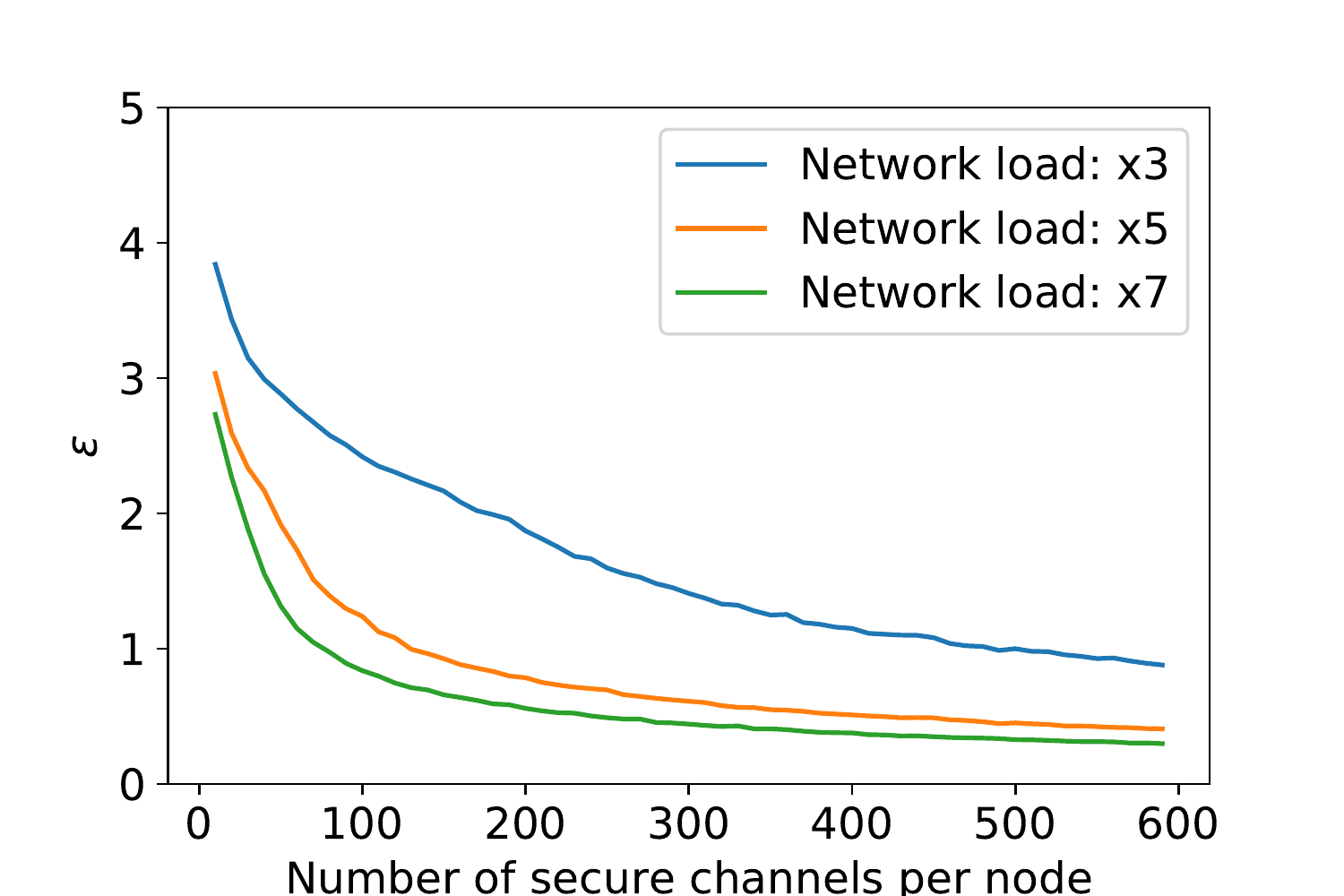}
    \caption{Aggregate: privacy vs individual load}
    \label{fig:gb-individual-load}
\end{subfigure}
\begin{subfigure}{.32\textwidth}
    \includegraphics[clip, trim=0 0 0 1cm, 
                width=\textwidth,keepaspectratio]{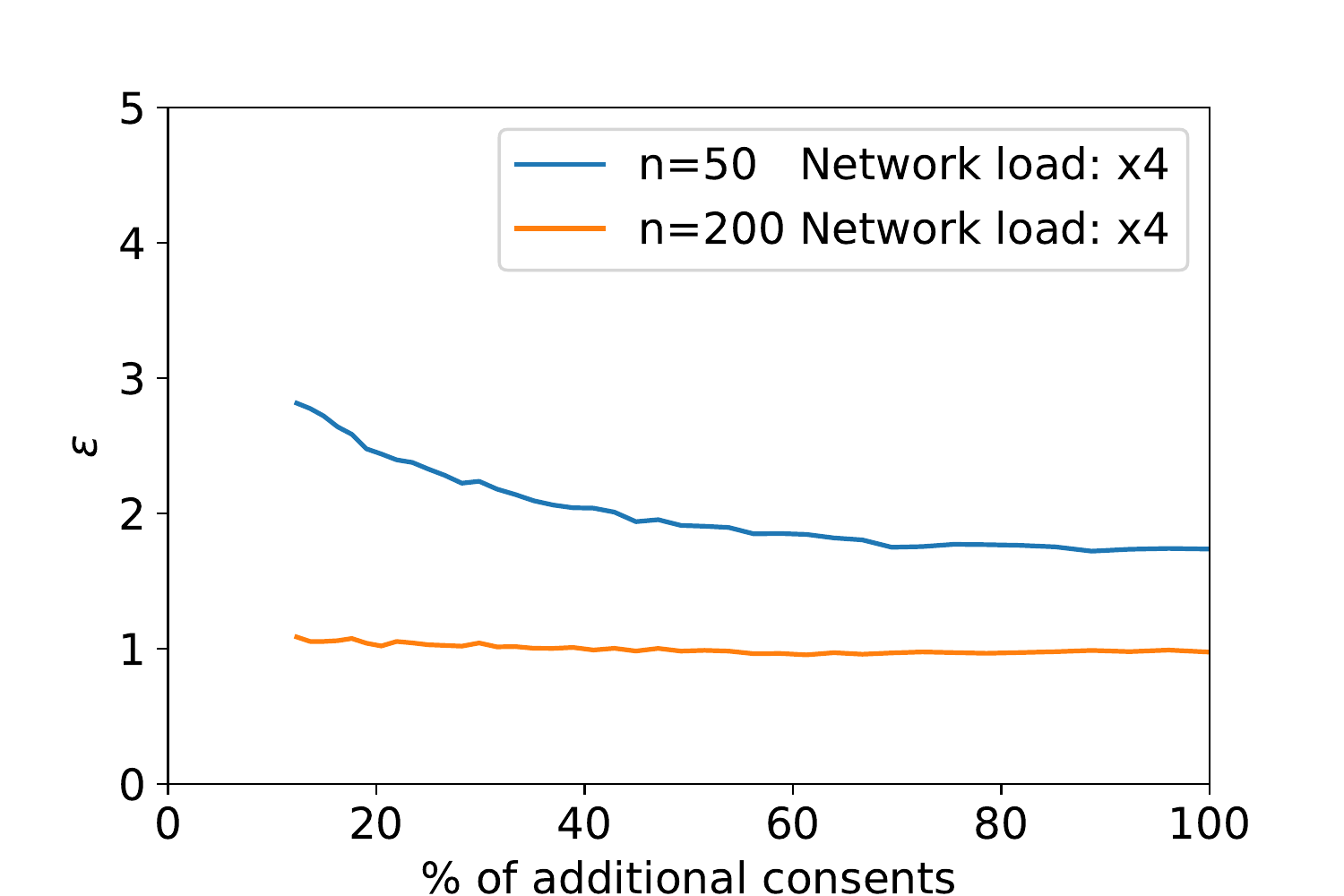}
    \caption{Aggregate: privacy vs consents}
    \label{fig:gb-consents}
\end{subfigure}

\begin{subfigure}{.32\textwidth}
    \includegraphics[clip, trim=0 0 0 1cm, 
                width=\textwidth,keepaspectratio]{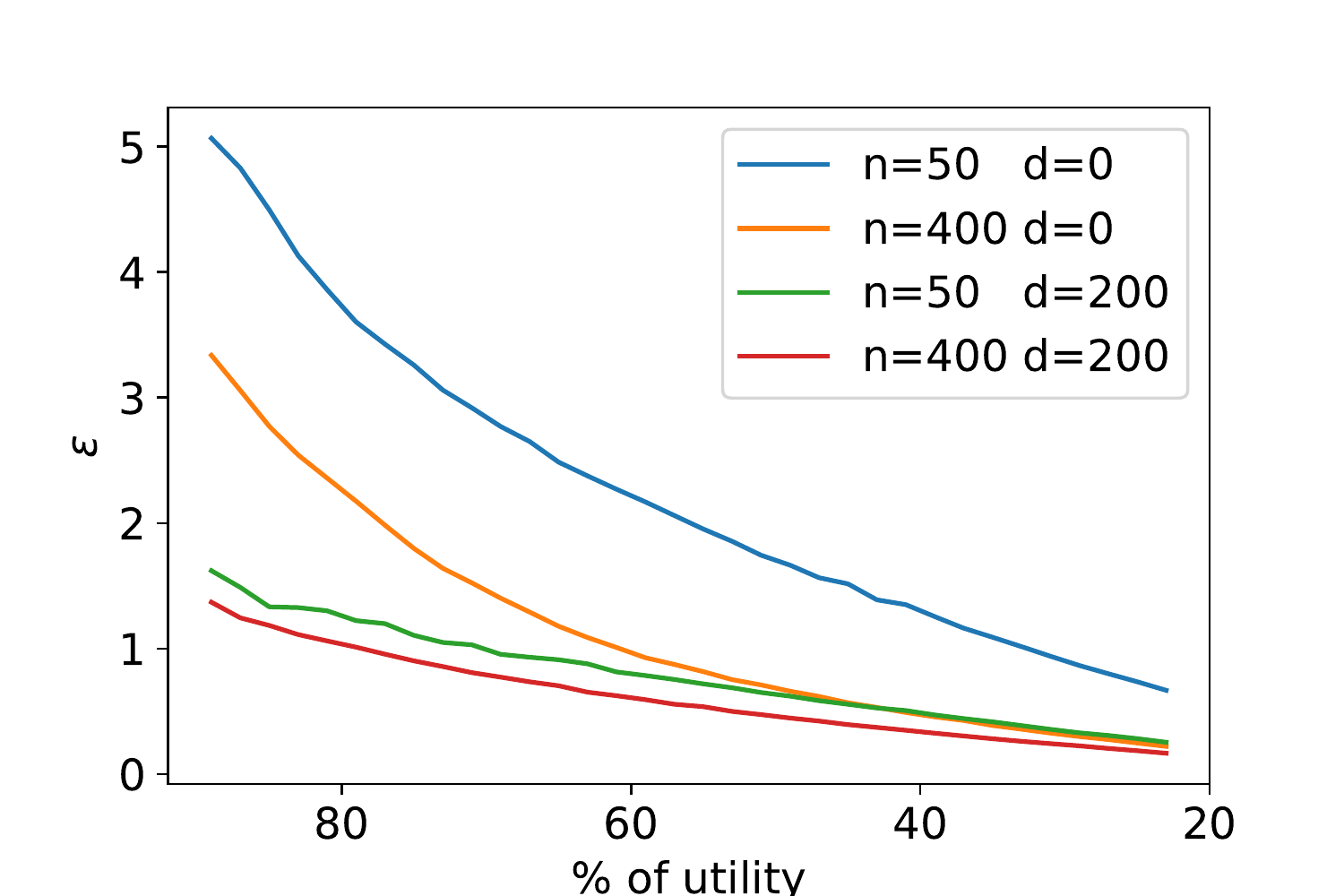}
    \caption{Aggregate: privacy vs utility}
   \label{fig:gb-utility}
\end{subfigure}
\begin{subfigure}{.32\textwidth}
    \includegraphics[clip, trim=0 0 0 1cm, 
                width=\textwidth,keepaspectratio]{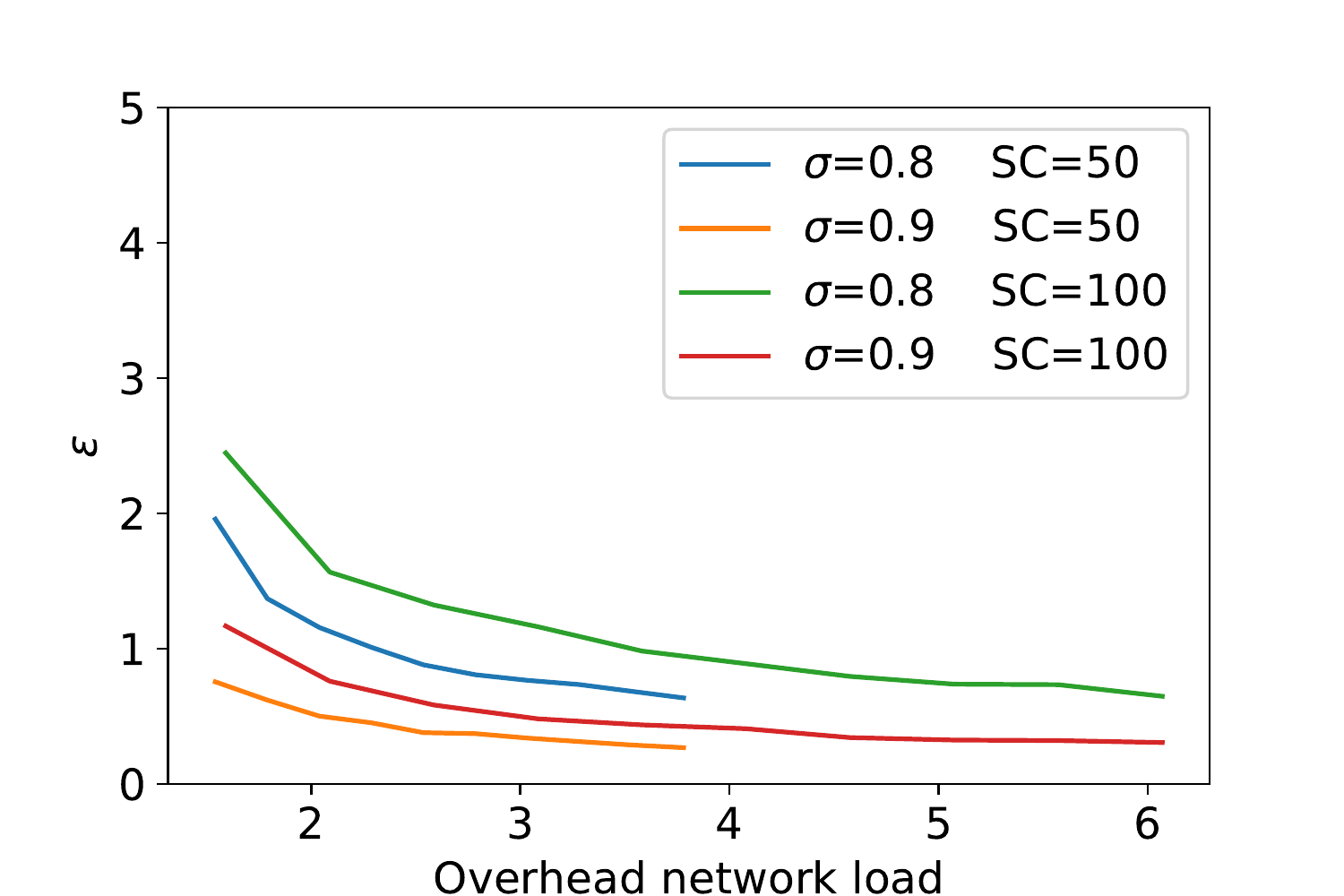}
    \caption{$K$-means: privacy vs global load}
    \label{fig:km-global-load}
\end{subfigure}
\begin{subfigure}{.32\textwidth}
    \includegraphics[clip, trim=0 0 0 1cm, 
                width=\textwidth,keepaspectratio]{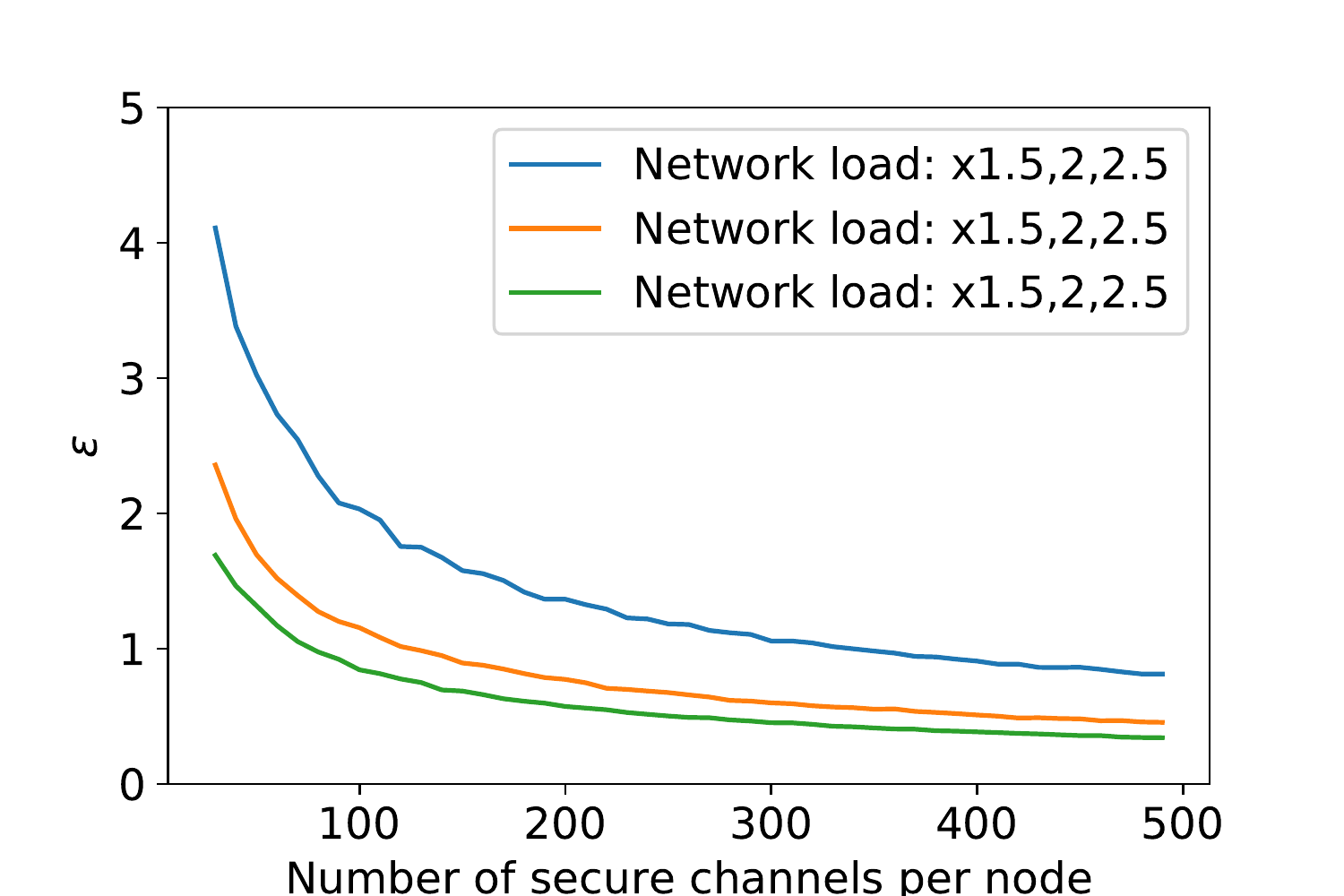}
    \caption{$K$-means: privacy vs individual load}
    \label{fig:km-individual-load}
\end{subfigure}

\caption{Trade-offs between privacy, utility and efficiency in Aggregate and $K$-means.}
\label{fig:tradeoffs}
\vspace{-1em}
\end{figure*}

The curves shown in \Cref{fig:tradeoffs,fig:figG}
 show our results. In each curve, we evaluate the trade-offs between efficiency/utility parameters (X-axis) and privacy (Y-axis). 
 
\emph{Network load vs privacy (\Cref{fig:gb-global-load,fig:km-global-load,fig:figG})}. 
We vary the number $d$ of dummies introduced by the scrambler nodes, fixing the number of tuples effectively used to $10k$, for different configurations. We consider low sampling rates for Aggregate (to maximize utility) and high sampling rates for $K$-Means (without utility loss). Network overhead exists even without dummy (left side of the curves) due to the introduction of scrambler nodes. When $d$ is increased, the privacy gains are significant, especially for the first dummies. Configurations with good privacy ($\epsilon \leq 1$) and acceptable network load can be achieved.

\begin{table}%
\parbox{0.15\textwidth}{
        \begin{tabular}{|C{0.8cm}|C{2.1cm}|}
         \hline
         Name & Range \\ \hline
        $C$ & 20 ($\times G$ $\times I$) \\
        $\mathrm{SC}$ & 20-100 ($\times G$ $\times I$) \\
        $I$ & 10 \\
        $G$ & 1-4 \\
        $S$ & $10k-20k$ \\
        $\epsilon$& $0-5$\\
        $\delta$ & $10^{-4}$\\
        $n$ & 10-600 \\
        $d$ & 0-1000 \\
        $\sigma$ & $0-1$ \\
         \hline
        \end{tabular}
    \captionof{table}{\centering Range of parameters for measures.}
    \label{tab:experiments-param}
}
\qquad
\begin{minipage}[c]{0.33\textwidth}%
    \includegraphics[clip, trim=0 0 0 1cm, width=\textwidth]{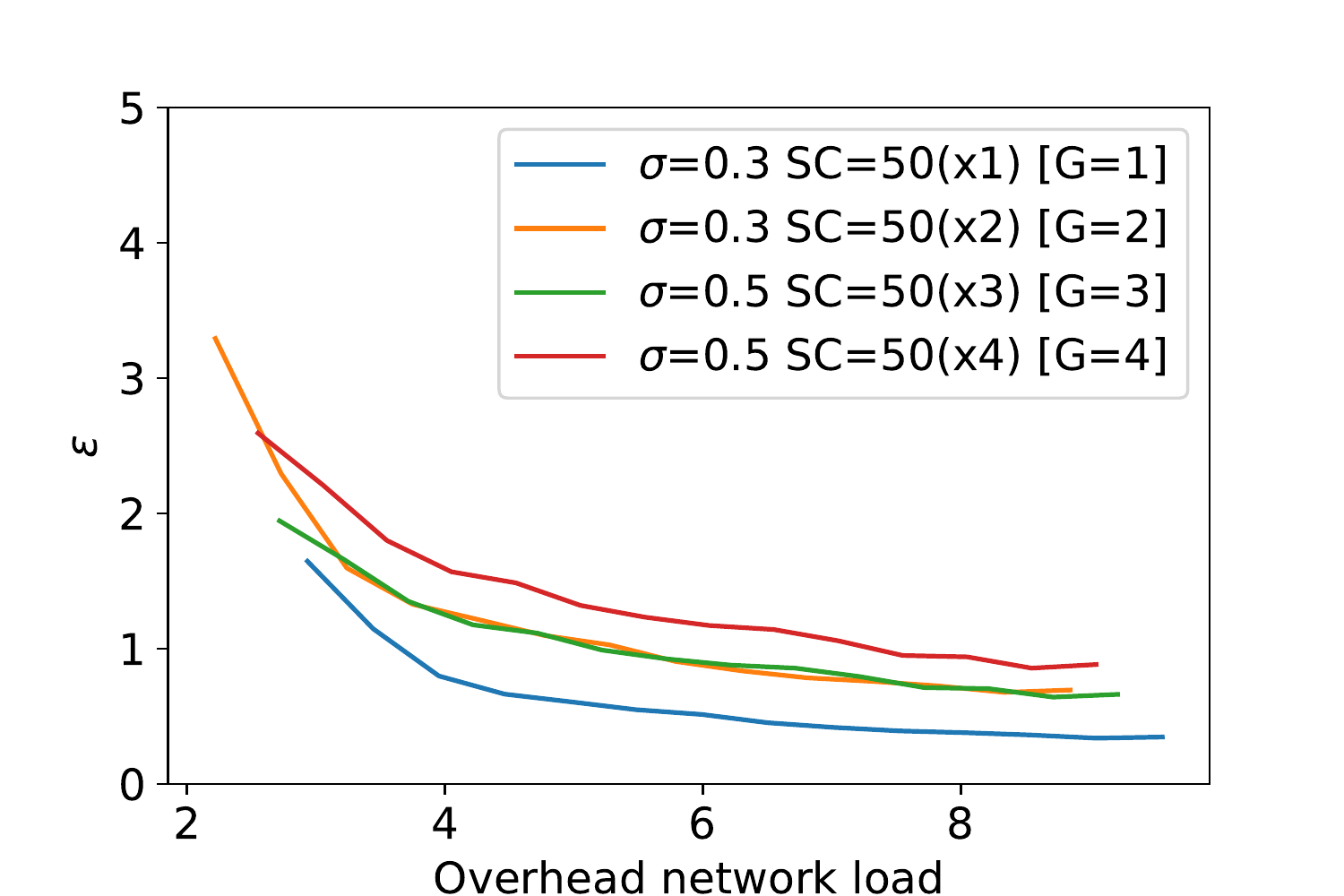}
\captionof{figure}{Aggr.: privacy vs net. load.}
\label{fig:figG}
\end{minipage}
\vspace{-2em}
\end{table}
\renewcommand{\arraystretch}{0.5} 
\setlength{\tabcolsep}{0.05cm} 

%
\emph{Additional users' consents vs privacy (\Cref{fig:gb-consents,fig:gb-utility})}. 
Fig.~\ref{fig:gb-consents} shows privacy at constant utility ($10k$ significant contributions) for fixed network overheads, increasing the sampling rate $\sigma$ hence the number $S$ of users' consents (up to $20k$), with the number $d$ of dummies kept compliant with the target network overhead. We observe that it is more privacy efficient to increase $\sigma$ (users' consents) than to increase $d$, especially with few messages per scrambler (low $n$, then higher $\mathrm{SC}$ and lower $d$). Similarly, Fig.~\ref{fig:gb-utility} shows that at fixed number of users' consents ($S$=10k) and dummies, sampling has a drastic effect on privacy, especially with fewer dummies.


\emph{Individual load vs privacy (\Cref{fig:gb-individual-load,fig:km-individual-load})}. 
Individual load is typically determined by the application context (i.e. limitations of individual participants).
\Cref{fig:gb-individual-load,fig:km-individual-load} show that increasing individual load on scramblers
yields better privacy, however this effect also diminishes when individual load increases. In particular, these figures suggest that a few tens to a few hundreds channels are enough in most cases.


Overall, depending on the constraints of the use case (acceptable utility loss, maximum individual load, \ldots), a number of good configurations can be reached by tuning the number of scramblers, the sampling rate and the number of dummies.

\nicolas{TODO : compléter avec des refs sur les algos listés et leur evaluation en
  distribué (data dependance pour raisons de perfs). Parler aussi des algos
  parallels de jointure (?). Ca peut aider aussi à  montrer que le cas de S sources qui distribuent des tuples à T targets est un cas général } 

\aurelien{use-cases ML dans lesquels on aurait de la communication data-dependent. Voici quelques éléments,
ça me permet à moi aussi de récapituler l'ensemble :\\

1/ clustering : bon candidat, tâche super classique en ML et analyse de données en général, et qui peut servir comme première étape ou composante dans du ML fédéré (voir 3/ ci-dessous). Les algos centroid-based (k-means, k-medians, k-medoids) semblent les plus adaptés à notre cas. De manière générale ces algos fonctionnent en alternant affectation des points au centroide le plus "proche" (cette notion peut varier entre les algos) et mise à jour des C centroides à partir des points qui leur sont affectés. Dans un contexte distribué comme le nôtre, il est naturel de procéder à l'affectation des points en local par les sources et à affecter une target à chaque centroid qui reçoit les points qui y sont affectés à chaque itération, met à jour le centroide et le renvoie aux sources. 

Voici quelques justifications :
\begin{itemize}
    \item L'adversaire peut déduire des communications les points qui sont
associés au même centroide (donc sont proches). S'il connait les
cendroides (ou certains points qui y sont affectés), il peut carrément
avoir une bonne idée d'où se trouve un point.
    \item Une stratégie de distribution alternative (mentionnée par Nico) est
de splitter les sources entre les targets de manière data-independent:
mais dans ce cas, on a besoin d'un round de communication supplémentaire
(pour aggréger les centroides partiels) ou bien de C fois plus de
communication (si on envoie directement les centroides partiels aux
sources). On peut aussi noter que cette approche est impossible pour
k-means et k-medoids, lequel le calcul de la mise à jour d'un centroide
nécessite l'accès à tous les points

\item Exemple avec MapReduce: dans le papier tout simple et très cité
suivant, allouer un reducer par centroide est ce qui permet de
paralléliser au max
\url{https://www.researchgate.net/profile/Qing-He/publication/225695804_Parallel_K-Means_Clustering_Based_on_MapReduce/links/5768a0f508ae8ec97a424884/Parallel-K-Means-Clustering-Based-on-MapReduce.pdf}

\item Autre exemple avec k-medoids, qui est particulièrement complexe à
paralléliser car la mise à jour classique d'un centroide nécessite de
regarder tous les points (pas seulement ceux affectés à ce centroide):
des versions parallèles proposent de faire simplement un "local search",
càd de ne considérer que les points affectés \\
\url{ http://article.nadiapub.com/IJSIP/vol7_no2/13.pdf}\\
\url{https://ieeexplore.ieee.org/document/6926527}

\end{itemize}

2/ On peut citer évoquer un autre algo de (density-based) clustering,
DBSCAN. J'ai regardé le papier évoqué par Nico la dernière fois:
\url{http://dm.kaist.ac.kr/jaegil/papers/sigmod18.pdf}

Les approches précédentes de DBSCAN distribué (citées dans l'article
ci-dessus) utilisent des splits de données data-dependent (par contre
ils sont eux-même coûteux à calculer : ça demande un preprocessing).
L'approche RP-DBSCAN proposée utilisent du pseudo-random splitting : en
gros ils découpent l'espace en cellule : les points qui tombent dedans
font partie du même split, et chaque split est un ensemble aléatoire de
cellules. Dans ce cas pas besoin de preprocessing et les communications
dépendent bien des données, même si cela peut révéler un peu moins
d'infos que dans le cas des algos de centroides par exemple.

3/ Le clustering permet de faire une transition naturelle vers des
approches récentes en Federated Learning qu'on appelle "Clustered
Federated Learning". L'idée est d'apprendre un modèle (type réseau de
neurones ou autre) mais au lieu d'un apprendre un seul pour tous les
clients, on va essayer d'apprendre un modèle par cluster de clients.
Cela permet d'avoir des modèles plus personnalisés qui sont adaptés à
des données/tâches hétérogènes entre clients. Dans ce cas, l'idée de
paralléliser en faisant qu'un serveur (target) soit chargé de la mise à
jour des modèles d'un cluster en aggrégeant les updates des clients
associés à ce cluster est hyper naturelle dans plein d'algos existants :
\begin{itemize}
    \item Dans certaines approches, le clustering des clients est fait comme un
preprocessing, par exemple en clusterisant les modèles appris localement
par les clients. Du coup il y a potentiellement des coms data-dependent
pour l'étape de clustering (par ex le 1er papier ci-dessous fait du
k-means), mais surtout une fois le clustering effectué, l'apprentissage
devient indépendant par cluster, donc "embarassingly parallel" avec des
coms data-dependent. \url{https://arxiv.org/pdf/1906.06629.pdf} \url{https://arxiv.org/pdf/1910.01991.pdf}

        \item Dans d'autres approches, le clustering des clients est effectué et
mis à jour en même temps que l'apprentissage des modèles. Les papiers
ci-dessous utilisent un seul serveur qui aggrège tout, mais en pratique
il est là encore super naturel de paralléliser en laissant un serveur
s'occuper de chaque cluster (càd recevoir les updates des clients
associés à ce cluster, mettre à jour le centroide, et le renvoyer aux
clients): \url{https://arxiv.org/pdf/2002.10619.pdf} \url{https://arxiv.org/pdf/2005.01026.pdf}
\url{https://arxiv.org/pdf/2006.04088.pdf}
\end{itemize}

4/ Moins lié au clustering, il y a des approches de Federated Learning
qui proposent de mieux passer à l'échelle d'un grand nombre de clients
en les organisant en sous-groupes qui bossent en parallèle et dont les
updates sont aggrégées par groupe avant d'être aggrégée entre groupes.
Ces groupes sont typiquement data-dependent afin d'assurer idéalement
une distribution des données similaire dans chaque groupe pour favoriser
la convergence, voir par exemple:
\url{https://arxiv.org/pdf/2012.03214.pdf}

5/ Un dernier point que je n'ai pas trop creusé, c'est l'apprentissage
de decision trees (algos de type ID3, C4.5, CART). A priori la
construction d'un decision tree nécessite de répondre à une séquence
adaptive de requêtes de type "quelle est la proportion des points de
chaque classe qui ont la valeur v pour l'attribut a?". Cela ressemble
donc fort au type de requêtes dont vous avez l'habitude (c'est un group
by ou un truc du style?). Je me dis donc que distribuer le calcul avec
des coms data-dependent pourrait être pertinent là aussi.

Conclusion: à mon avis on a des billes pour mettre en avant des
use-cases ML. Il faut voir comment le rédiger de manière claire (je
pourrai évidemment m'en charger). Si on veut des expériences, il vaut
mieux sans doute s'en tenir à un bon vieux k-means pour que ce soit pas
trop compliqué :-)

== Fin mail}

\section{Related Work} 
\label{sec:rw}
\emph{Anonymous communications.}
Providing ways of communicating anonymously is far from a new problem. While our
appoach is focused on tackling data dependency in communications, it is closely
related to various lines of work seeking to hide endpoints of communications.

The first and probably simplest way of providing anonymous communications is
to use mix networks (mixnets) \cite{Mixnets}. Recent work has studied how shuffling messages with mixnets can amplify local differential privacy guarantees (see Section~\ref{sec:RW-DP}). While
we take inspiration from such works in our use of scramblers (which we do not
deliberately call mixnets or shufflers as they also have the function of adding dummy messages), simply
using mixnets would not provide differential privacy guarantees in
our case as they do not hide the number of messages sent to each
target. Similar solutions seeking to achieve anonymous routing (e.g TOR
\cite{TOR}) have the same problem, and typically induce a
fairly high overhead in terms of communications and cryptographic
computations at client side.

Differentially private messaging is more closely related to our goal. Vuvuzela \cite{vuvuzela} and a number of following
works \cite{karaoke} seek to provide differentially private communications. These would
satisfy our goal of hiding data dependency in a differentially private
manner. However, these aim at a larger goal, which is to fully hide who communicates with
who (and even the fact that users are communicating at all) rather than
restricting the problem to hiding data dependency. In our case, we are willing to
disclose the fact that a source is talking to a target, we simply want to hide
which specific target. As a consequence, these works have a very high overhead
as they need to drown all actual traffic within fake traffic (the system should
behave in roughly the same way whether people are communicating or not), leading
to network load being orders of magnitude greater than the number of actual
messages sent. In our work we leverage the fact that distributed
computations typically do not exhibit arbitrary communication patterns to obtain a
much more efficient solution.

Finally, the work of \cite{AnoA} aims at modeling and providing tools for
analyzing protocols where cryptographic guarantees and differential guarantees
coexist and providing differential privacy notions suitable for anonymous communications. Our adversary model is largely inspired by this work: in particular, our reduction from computational differential privacy and the adjacency notion for communication graphs in clusters is similar to the one in \cite{AnoA}.

\nicolas{todo gs : (1) ajouter les refs; (2) vérifier qu'il ne manque rien, par ex: Prochlo, karaoke, voir aussi :
 Bittau, A., Erlingsson, Ú., Maniatis, P., Mironov, I., Raghunathan, A., Lie, D., ... Seefeld, B. (2017, October). Prochlo: Strong privacy for analytics in the crowd. In Proceedings of the 26th Symposium on Operating Systems Principles (pp. 441-459).
 Wang, C., Bater, J., Nayak, K.,  Machanavajjhala, A. (2021, June). DP-Sync: Hiding Update Patterns in Secure Outsourced Databases with Differential Privacy. In Proceedings of the 2021 International Conference on Management of Data (pp. 1892-1905).
 Lazar, D., Gilad, Y.,  Zeldovich, N. (2018). Karaoke: Distributed private messaging immune to passive traffic analysis. In 13th {USENIX} Symposium on Operating Systems Design and Implementation ({OSDI} 18) (pp. 711-725).

A voir aussi sur le lien crypto + DP: 
Sameer Wagh, Xi He, Ashwin Machanavajjhala, Prateek Mittal:
DP-cryptography: marrying differential privacy and cryptography in emerging applications. Commun. ACM 64(2): 84-93 (2021)
https://dl.acm.org/doi/pdf/10.1145/3418290

}



\emph{Differentially private data analysis.}
\label{sec:RW-DP}
In terms of techniques used in this paper our work is closely related to differential privacy in the shuffle model \cite{Cheu2019,amp_shuffling, Balle2019} which provides an intermediate model between central and local DP. The shuffle model can reduce the utility cost of the local model by passing randomized data points through a secure shuffler (mixnet) before they are shared with an untrusted third party. However, many queries do not admit accurate solutions in the shuffle model \cite{shuffle_limits}. In contrast we rely here on user-side trusted environments for distributed query evaluation, which provides different trade-offs. In particular, we avoid the loss in utility of local DP without requiring a trusted third party, and allow to accurately evaluate general queries (albeit at a potentially large cost in efficiency if trusted execution environments must be used to secure user-side computation).
An original aspect of our work is to leverage DP and amplification by shuffling to guarantee the privacy of \emph{data-dependent communications patterns} (and thereby mitigate traffic analysis attacks), while the above work on the shuffle and local models uses DP to guarantee the privacy of the \emph{content of messages} with data-independent communication patterns. 
We also stress the fact that our approach nicely composes with central DP in use-cases where the result of the query evaluated in our framework is released in a differentially private way. In particular, if we provide an $(\epsilon_1,\delta_1)$-DP guarantee for communication patterns and an $(\epsilon_2,\delta_2)$-DP guarantee for releasing the query result, then by the composition property of DP we obtain an $(\epsilon_1+\epsilon_2,\delta_1+\delta_2)$-DP guarantee against an adversary who observes both the communication patterns and the final result.

\emph{Hiding memory access patterns and input/output size.}
Many solutions were proposed to hide the query and/or the size of inputs and outputs of operators in execution plans \cite{bater2018shrinkwrap, bater2020saqe, xu2019hermetic, ren2020hybridx}. 
While the goal of these works is to hide data dependency in execution flows, these approaches tackle a different problem from ours. Indeed, only trees are considered (i.e. any operator has a single successor) and the goal is to hide the query executed within an operator and/or the result size, for example, hiding the number of tuples returned by a selection query. In contrast, our proposal supports any query plan, and prevents attackers from inferring information about individual tuples by observing the communications.

Another line of work consists in hiding memory access patterns in distributed cloud computations and local computations with multiple processors. In particular, \cite{DP-Orimenkho, DP-Mazloom} provide differential privacy guarantees for memory access patterns. Specifically, \cite{DP-Mazloom} offers modifications to secure two-party computations between two non-colluding servers that hold individual data, and \cite{DP-Orimenkho} offers a definition of oblivious differential privacy which is roughly the counterpart of ours when considering memory access pattern rather than communications but focuses on a single process rather than on interactions between multiple processes.
Interestingly, these approaches can be used to protect individual nodes (which may perform complex computations), and the differential privacy guarantees they provide would compose nicely with ours.

\emph{Hiding data dependency in communication patterns.}
\label{sec:RW-side}
Several existing works use various anonymous communication techniques to hide data exchange between nodes in distributed query plans in a cloud setting \cite{Prochlo, M2R, zheng2017opaque}. While these are somewhat related to ours, they assume users have a (very) large amount of data and offer either unsatisfactory guarantees in our massively decentralized setting or massive overheads.


A more direct way to hide the dependency between communication patterns and private data values is to make all communications data-independent, as done in \cite{sh:observing} by padding and clipping messages in MapReduce computations for confidential computing in the cloud. However, this technique would produce massive overheads or very imprecise results in our case, where there is no prior knowledge on the data distribution to appropriately tune the padding and clipping parameters.

\section{Conclusion} 
\label{sec:conclusion}

In this paper, we proposed a differentially private solution to mitigate the leakage from data-dependent communications in massively distributed computations. We leveraged recent work on privacy amplification by shuffling to formally prove privacy guarantees for our solution. We also showed how to balance privacy, utility and efficiency on two use-cases representative of distributed computations, highlighting the genericity of our solution.


We hope that our proposal will contribute to the development of new decentralized models for Data Altruism \cite{EUDataGovAct2020}, in which citizens contribute the computation of socially useful information, with community control over the computation performed on the user side.
Many research questions remain open, such as formulating and validating a collective computing ``manifesto''. Another concrete challenge relates to the implementation of a platform to support these technologies. An interesting future work is to build on the emergence of Personal Data Management Systems (PDMS) \cite{DBLP:journals/is/AnciauxBBNPPS19, urquhart2018realising}, which provide new tools for individuals to collect their personal data and control how they share results of local computations. 


\section*{Acknowledgments}

This work was supported by the French National Research Agency (ANR) through grants ANR-16-CE23-0016 (Project PAMELA) and ANR-20-CE23-0015
(Project PRIDE).

\bibliographystyle{plain}
\bibliography{references}

\clearpage
\appendices

\section{Computational differential privacy}
\label{app:comp-diff-priv}
We explain how we can reduce our attention to the privacy notion defined in Definition~\ref{def:dp1} under our hypotheses. Similarly to \cite{ComputationalDP}, we start from a notion which explicitly restricts the power of the adversary.

\begin{definition}[Computational DP for execution plans]
\label{def:dp0}
  An execution plan $\mathcal{N}$ is $(\epsilon, \delta)$-differentially private
  if for any neighboring $\mathcal{D}_0, \mathcal{D}_1$ (differing by at most one
  tuple), and any probabilistic polynomial time adversary $\mathcal{A}$ we have:
  \begin{multline*}
    P[\mathcal{A} \text{ interacting with }\mathcal{N}(\mathcal{D}_0) \text{ guesses } 0] \\
    \leq e^\epsilon P[\mathcal{A} \text{ interacting with }\mathcal{N}(\mathcal{D}_1) \text{ guesses } 0] + \delta.
  \end{multline*}
\end{definition}

We can make two immediate simplifications to this definition using our assumption that communication channels are secure. The first one is to exclude adversaries that
influence the computation by injecting messages (as such an adversary would
need to break secure channel integrity). We can therefore restrict ourselves
to a passive adversary. The second one is that we may also omit the content of
messages, as under secure channels all messages can be replaced
by random messages without the adversary being able to differentiate the two
situations. Finally, as two identical communication graphs are trivially
indistinguishable when the content of messages is random, the question reduces to
whether different neighboring datasets generate similar communication
graphs, leading to our Definition~\ref{def:dp1} in the main text.


\section{Proof of \Cref{thm:privacy_local}}
\label{app:proof-priv-local}

\begin{proof}
Let $\mathcal{G}$ and $\mathcal{G}'$ be two neighboring communication graphs and let $(s,t)$ and $(s,t')$ with $t\neq t'$ be the messages that differ in $\mathcal{G}$ and $\mathcal{G}'$ respectively. To prove that $\Algo_{\Plie}^{d}$ satisfies $\epsilon$-differential privacy, we need to show that for any possible output $\mathcal{O}$:
$$\frac{P[\Algo_{\Plie}^{d}(\mathcal{G})=\Output]}{P[\Algo_{\Plie}^{d}(\mathcal{G}') = \Output]} \leq e^\epsilon.$$
Since each message is processed independently by $\Randomizer_{\Plie}^{d}$, and $\mathcal{G}$ and $\mathcal{G}'$ are neighboring, we have:
\begin{equation}
\label{eq:A_R}
\frac{P[\Algo_{\Plie}^{d}(\mathcal{G})=\Output]}{P[\Algo_{\Plie}^{d}(\mathcal{G}') = \Output]} =\frac{P[\Randomizer_{\Plie}^{d}(s,t)=\Output]}{P[\Randomizer_{\Plie}^{d}(s,t') = \Output]}.
\end{equation}

We now seek to bound the above ratio for the worst-case output. By construction of $\Randomizer_\Plie^d$, the probability of producing a message $(s,t'')$ with $t''\neq t\neq t'$ is the same for $\Randomizer_\Plie^d(s,t)$ and $\Randomizer_\Plie^d(s,t')$. Therefore, we can focus on four cases for the output $\Output$ produced by $\Randomizer_{\Plie}^{d}$: (i) $(s,t) \in \Output$ and $(s,t') \in \Output$, (ii) $(s,t) \notin \Output$ and $(s,t') \notin \Output$, (iii) $(s,t)\in \Output$ and $(s,t') \not\in \Output$ and (iv) $(s,t) \not\in \Output$ and $(s,'t) \in \Output$.

The probability of the first two cases is the same under both inputs, hence their ratio is 1. The two last cases are symmetric, so without loss of generality we consider an output $\Output$ such that $(s,t) \in \Output$ and $(s,t') \not\in \Output$.
Consider sampling $d+1$ elements without replacement from $\Target$: for two distinct $t,t'\in \Target$, let $E_{\neg t'}$ be the event where $t'$ is not selected, $E_{\neg t \wedge \neg t'}$ the event where neither $t$ nor $t'$ are selected and $E_{t \wedge \neg t'}$ the event where $t$ is selected but not $t'$. We have:
\begin{align*}
P[E_{t \wedge \neg t'}] &= P[E_{\neg t'}] - P[E_{\neg t \wedge \neg t'}] \\
& = \frac{T-d-1}{T}-\frac{(T-d-1)(T-d-2)}{T(T-1)}\\
&= (d + 1) \frac{T- d - 1}{T(T-1)}.
\end{align*}
Using the above, we can compute the ratio of probabilities in Eq.~\ref{eq:A_R}:
\begin{align}
\frac{P[\Randomizer_{\Plie}^{d}(s,t)=\Output]}{P[\Randomizer_{\Plie}^{d}(s,t') = \Output]} &=\frac{(1-\Plie) \frac{T-(d+1)}{T-1}  + \Plie P[E_{t \wedge \neg t'}]}{\Plie  P[E_{t \wedge \neg t'}]}\nonumber\\
&= \frac{(1-\Plie) \frac{T-(d+1)}{T-1} + \Plie (d + 1) \frac{T-d - 1}{T(T-1)} }{\Plie  (d + 1)  \frac{T- d - 1}{T(T-1)}}\nonumber\\
&= \frac{(1-\Plie)T}{\Plie ( d + 1 )}+1\leq e^\epsilon,\label{eq:ldp}
\end{align}
which combined with Eq.~\ref{eq:A_R} shows that $\Algo_{\Plie}^{d}$ satisfies $\epsilon$-DP. 
\end{proof}
Incidentally, Eq.~\ref{eq:ldp} shows that the local randomizer $\Randomizer_{\Plie}^{d}$ satisfies $\epsilon$-local differential privacy.

\section{Proof of Theorem~\ref{theorem-exact}}
\label{app:Appendix1}

To facilitate the reading, we introduce additional notations. Let $\mathcal{G}$ and $\mathcal{G}'$ be two neighboring communication graphs and let $(s,t)$ and $(s,t')$ with $t\neq t'$ be the messages that differ in $\mathcal{G}$ and $\mathcal{G}'$ respectively. We abstract the graphs $\mathcal{G}$ as $\mathcal{G}=\{ x, y, z\}$ where $x$ is the number of messages targeting $t$, $z$ is the number of messages targeting $t'$ and $y$ is the number of messages targeting $t'' \in \Target \setminus \{t,t'\}$  with $t \neq t'$. Following the same principle we denote the output as $\Output=\{ \alpha, \beta, \gamma \}$ where $\alpha$ is the number of messages the scrambler sends to the target $t$, $\gamma$ is the number of messages are sent to the target $t'$ and $\beta$ is the number of messages are sent to $t'' \in \Target \setminus \{t,t'\}$.
The other used notations are summarized below:

\renewcommand{\arraystretch}{0.3} 
\setlength{\tabcolsep}{0.05cm} 

\begin{itemize}

    \item $\PND{\alpha}{\beta}{\gamma}{x}{y}{z}$ is the probability to get an output $\Output=\{ \alpha, \beta, \gamma \}$ with the algorithm $\Abar_\Plie^{n,d}$ given a graph $\mathcal{G}=\{ x, y, z\}$ 
  
    \item $R_{n}^{d} \Inn{\alpha}{\beta}{\gamma} {x}{y}{z}$ is the ratio $\frac{ \PND{\alpha}{\beta}{\gamma}{x\!+\!1}{y}{z}} {\PND{\alpha}{\beta}{\gamma}{x}{y}{z\!+\!1}}$ and it is equal to $e^\epsilon$
    
    \item $\PNDum{k_1}{k_2}{k_3}{\alpha-k_1}{\beta-k2}{\gamma-k_3}$ is the probability to send $k_1$ dummies (resp. $k_2$, $k_3$) to the target $t$ (resp. $t'' \in \Target \setminus \{t,t'\}$, $t')$.
    
    \item $\Phi_{u,v} =$ is the probability to draw $u$ times $\alpha$, $v$ times $\gamma$ and $x+z-u-v$ times $\beta$ (i.e $\PNZ{u}{x+z-u-v}{v}{x}{0}{z}$).
    
    \item $\I(\text{conditions})$ is the indicator function (equal to 1 when the conditions are met and 0 otherwise).
\end{itemize}

\paragraph{Useful formulas} based on the notations above, we provide below some useful formulas used in the remaining of the chapter.

\begin{equation}
\PND{\alpha}{\beta}{\gamma}{x}{y}{z}
=\!\!\!\!\!\! \underset{k_1+k_2= d}{\sum} \PNDum{k_1}{\!d\!-\!k_1\!-\!k_2}{k_2}{\alpha-k_1}{\!\beta\!-\!(\!d\!-\!k_1\!-\!k_2\!)\!}{\gamma-k_2} \cdot \PNZ{\alpha-k_1}{\beta-(d-k_1-k_2)}{\gamma-k_2}{x}{y}{z}
\label{eq:tirage_dummies}
\end{equation}

\begin{equation}
\PNZ{\alpha}{\beta}{\gamma}{0}{y}{0}
= \frac{y!{\truth}^{\beta} \lie^{y-\beta}}
       {\alpha!\beta!\gamma!} 
  \I\Ink{\alpha \geq 0}{\beta \geq 0}{\gamma \geq 0}
\label{eq:p_n_0_y}
\end{equation}

\begin{equation}
\PNZ{\alpha}{\beta}{\gamma}{x}{y}{z}= \sum_{u+v \leq x+z}^{}{ \Phi_{u,v}
\PNZ{\alpha - u}{\beta - (x+z-u-v)}{\gamma -v)}{0}{y}{0}}
\label{eq:p_n_0_xyz}
\end{equation}

By replacing formula (\ref{eq:p_n_0_y}) in formula (\ref{eq:p_n_0_xyz}) we obtain:

\begin{equation}
\begin{aligned}
\PNZ{\alpha}{\beta}{\gamma}{x}{y}{z}
=\!\!\!\!\!\! \underset{u + v \leq x + z}{ \sum}\!\!\!\!\!\! & \Phi_{u,v} 
  \frac{y! {\truth}^{\beta-v} \lie^{y-\beta+v}} 
       {(\alpha-u)!(\beta-(x+z-u-v))!(\gamma-v)!}\\&\times
        \I\Ink{u \leq \alpha} {\!\! x\!\! +\!\! z\!\! -\!\! u\!\! -\!\! v\!\! \leq \!\! \beta\!\!} {v \leq \gamma}
\end{aligned}
\label{eq:p_n_0_x_y_z}
\end{equation}

\begin{proof}
We first need to determine the input and the output producing the higher ratio\footnote{Ignoring the symmetric case where the ratio is minimum.} for two neighboring graphs. In other words we need to find $\mathcal{G}=\{x+1, y, z\}$, $\mathcal{G}'=\{x, y, z+1\}$ and $\Output=\{\alpha, \beta, \gamma \}$
such that $R_{n}^{d} \Inn{\alpha}{\beta}{\gamma} {x}{y}{z}=\frac{P[\Algo(\mathcal{G}) \in \Output]} {P[\Algo(\mathcal{G}') \in \Output)]}$ is maximum.\\

We have: 
\begin{equation*}
R_{n}^{d} \Inn{\alpha}{\beta}{\gamma} {x}{y}{z}
= \frac{ \PND{\alpha}{\beta}{\gamma}{x\!+\!1}{y}{z}} {\PND{\alpha}{\beta}{\gamma}{x}{y}{z\!+\!1}} 
\end{equation*}

We start by developing the numerator: 

\begin{equation*}
    \PND{\alpha}{\beta}{\gamma}{x\!+\!1}{y}{z}=\underset{k_1+k_2= d}{\sum} \PNDum{k_1}{d\!\!-\!\!k_1\!\!-\!\!k_2}{k_2}{\alpha-k_1}{\beta\!\!-\!\!(\!d\!\!-\!\!k_1\!\!-\!\!k_2\!)}{\gamma\!\!-\!\!k_2} \cdot \PNZ{\alpha\!\!-\!\!k_1}{\beta\!\!-\!\!(\!d\!\!-\!\!k_1\!\!-\!\!k_2\!)}{\gamma\!\!-\!\!k_2}{x\!\!+\!\!1}{y}{z}
\end{equation*} 

\begin{align*}
\PND{\alpha}{\beta}{\gamma}{x\!+\!1}{y}{z}=\underset{k_1+k_2= d}{\sum}& \PNDum{k_1}{d\!-\!k_1\!-\!k_2}{k_2}{\alpha-k_1}{\beta\!-\!(\!d\!-\!k_1\!-\!k_2\!)}{\gamma\!-\!k_2}\\
&\cdot \Bigg( \truth\PNZ{\alpha\!-\!k_1\!-\!1}{\beta\!-\!(\!d\!-\!k_1\!-\!k_2\!)}{\gamma\!-\!k_2}{x\!}{y}{z}\\
&+\lie\PNZ{\alpha\!-\!k_1}{\beta\!-\!(\!d\!-\!k_1\!-\!k_2\!)\!-\!1}{\gamma\!-\!k_2}{x\!}{y}{z}\\
&+\lie\PNZ{\alpha\!-\!k_1}{\beta\!-\!(\!d\!-\!k_1\!-\!k_2\!)}{\gamma\!-\!k_2\!-\!1}{x\!}{y}{z}\Bigg)
\end{align*}

We then apply the formula \ref{eq:p_n_0_x_y_z} to each $\mathbb{P}_{n,0}$:

\begin{footnotesize}
\begin{align*}
\PNZ{\alpha\!-\! k_1\!-\!1}{\beta\!-\!(\!d\!-\!k_1\!-\!k_2\!)}{\gamma\!-\!k_2\!}{x}{y}{z}
=
\end{align*}
\begin{align*}
\underset{u + v \leq x  +z}{ \sum}\!\!\!\!\!  &
  \frac{y! {\truth}^{\beta\!-\!(\!d\!-\!k_1\!-\!k_2\!)-v} \lie^{y-\beta\!-\!(\!d\!-\!k_1\!-\!k_2\!)+v}}
       {(\alpha\!-\! k_1\!-\!1-u)!(\beta\!-\!(\!d\!-\!k_1\!-\!k_2\!)\!-\!(x+z-u-v))!(\gamma\!-\!k_2\!-v)!}\\ 
  &\cdot \Phi_{u,v} \cdot \I\Ink{u \leq \alpha\!-\! k_1\!-\!1} {\!\! x\!\! +\!\! z\!\! -\!\! u\!\! -\!\! v\!\! \leq \!\! \beta\!-\!(\!d\!-\!k_1\!-\!k_2\!)\!\!} {v \leq \gamma\!-\!k_2\!} 
\end{align*}
\begin{align*}  
  \PNZ{\alpha\!-\! k_1\!-\!1}{\beta\!-\!(\!d\!-\!k_1\!-\!k_2\!)}{\gamma\!-\!k_2\!}{x}{y}{z}&= \underset{u + v \leq x + z}{ \sum} f\Ink{\alpha,k_1}{\beta,k_2}{\gamma,k_3} (\alpha\!-\!k_1\!-\!u) \I\Ink{\alpha \!-\! k_1\! -\! u\! \geq\! 1}{1}{1}
\end{align*}
\end{footnotesize}
where:
\begin{footnotesize}
\begin{align*} 
 f\Ink{\alpha,k_1}{\beta,k_2}{\gamma,k_3}=& \frac{y!{\truth}^{\beta\!-\!(\!d\!-\!k_1\!-\!k_2\!)-v} \lie^{y-\beta\!-\!(\!d\!-\!k_1\!-\!k_2\!)+v}} {(\alpha\!-\! k_1\!-u)!(\beta\!-\!(\!d\!-\!k_1\!-\!k_2\!)-(x+z-u-v))!(\gamma\!-\!k_2\!-v)!} \\
 &  \Phi_{u,v} \cdot \I\Ink{u \leq \alpha\!-\! k_1\!} {\!\! x\!\! +\!\! z\!\! -\!\! u\!\! -\!\! v\!\! \leq \!\! \beta\!-\!(\!d\!-\!k_1\!-\!k_2\!)\!\!} {v \leq \gamma\!-\!k_2\!} \\
\end{align*}
\end{footnotesize}

In the same way we obtain for the two other $\mathbb{P}_{n,0}$: 

\begin{small}
\begin{align*}
\PNZ{\alpha\!-\! k_1\!}{\beta\!-\!(\!d\!-\!k_1\!-\!k_2\!)-\!1}{\gamma\!-\!k_2\!}{x}{y}{z}
=
\end{align*}
\begin{align*}
\underset{u + v \leq x + z}{ \sum} f\Ink{\alpha,k_1}{\beta,k_2}{\gamma,k_3} 
  \frac{\lie (\beta\!-\!(\!d\!-\!k_1\!-\!k_2\!)-(x+z-u-v))}
       {\truth}\\
  \cdot \I\Ink{1}{(\beta\!-\!(\!d\!-\!k_1\!-\!k_2\!)-(x+z-u-v))-1\geq 0}{1}
\end{align*}

\begin{align*}
\PNZ{\alpha\!-\! k_1\!}{\beta\!-\!(\!d\!-\!k_1\!-\!k_2\!)}{\gamma\!-\!k_2\!-\!1}{x}{y}{z}
=\!\!\!\! \underset{u + v \leq x + z}{ \sum}\!\! f\Ink{\alpha,k_1}{\beta,k_2}{\gamma,k_3} (\gamma-\!k_2-v) \I\Ink{1}{1}{\gamma-\!k_2\!-\!v\!-\!1\!}
\end{align*}
\end{small}

By replacing in the numerator we obtain:
\begin{align*}
    \PND{\alpha}{\beta}{\gamma}{x\!+\!1}{y}{z}&=\underset{u+v\leq x+z}{\underset{k_1+k_2\leq d}{\sum}} \!\!\Omega\!\Ink{\alpha,k_1}{\beta,k_2}{\gamma,k_3} \!\cdot\! 
\Bigg(
    \!\!\truth(\alpha-\!k_1-u) 
       \I\Ink{\!\!\alpha \!-\!\! k_1\!\!-\!\! u\! \geq\! 1\!\!}{1}{1}\\
     &+ \frac{\lie^2 (\beta\!-\!(\!d\!-\!k_1\!-\!k_2\!)-(x+z-u-v))}
            {\truth} \\
         &\cdot \I\Ink{1}{\!\!\!(\!\beta\!\!-\!\!(\!d\!\!-\!\!k_1\!\!-\!\!k_2\!)\!\!-\!\!(\!x\!\!+\!\!z\!\!-\!\!u\!\!-\!\!v))\!\!\geq\!\!1\!\!}{1}\\
    & + \lie(\gamma-\!k_2-v) 
       \I\Ink{1}{1}{\!\!\gamma-\!k_2\!-\!v\!\geq\!1\!\!}
\!\!\Bigg)
\end{align*}

where $ \Omega\Ink{\alpha,k_1}{\beta,k_2}{\gamma,k_3}= \PNDum{k_1}{d\!-\!k_1\!-\!k_2}{k_2}{\alpha-k_1}{\beta\!-\!(\!d\!-\!k_1\!-\!k_2\!)}{\gamma\!-\!k_2} f\Ink{\alpha,k_1}{\beta,k_2}{\gamma,k_3} $.\\

With the same reasoning for the denominator, we obtain:

\begin{small}
\begin{align*}
    \PND{\alpha}{\beta}{\gamma}{x}{y}{z\!+\!1}&=\!\!\!\!\!\!\underset{u+v\leq x+z}{\underset{k_1+k_2\leq d}{\sum}} \!\!\!\!\Omega\!\Ink{\alpha,k_1}{\beta,k_2}{\gamma,k_3} \!\!
\Bigg(
    \!\!\lie(\alpha-\!k_1-u) 
       \I\Ink{\!\!\alpha \!-\!\! k_1\!\!-\!\! u\! \geq\! 1\!\!}{1}{1}\\
     &+ \frac{\lie^2 (\beta\!-\!(\!d\!-\!k_1\!-\!k_2\!)-(x+z-u-v))}
            {\truth}\\
        &\cdot \I\Ink{1}{\!\!\!(\!\beta\!\!-\!\!(\!d\!\!-\!\!k_1\!\!-\!\!k_2\!)\!\!-\!\!(\!x\!\!+\!\!z\!\!-\!\!u\!\!-\!\!v))\!\!\geq\!\!1\!\!}{1}\\
     &+ \truth(\gamma-\!k_2-v) 
       \I\Ink{1}{1}{\!\!\gamma-\!k_2\!-\!v\!\geq\!1\!\!}
\!\!\Bigg)
\end{align*}
\end{small}

The ratio can then be written as:

\begin{align*}
R_{n}^{d} \Inn{\alpha}{\beta}{\gamma} {x}{y}{z}
= \frac{\underset{u+v\leq x+z}{\underset{k_1+k_2\leq d}{\sum}} \!\!\Omega\!\Ink{\alpha,k_1}{\beta,k_2}{\gamma,k_3} \!\cdot\! 
\left(
    \truth \chi_1
     + \chi_2
     + \lie \chi_3
\right)
}
{\underset{u+v\leq x+z}{\underset{k_1+k_2\leq d}{\sum}} \!\!\Omega\!\Ink{\alpha,k_1}{\beta,k_2}{\gamma,k_3} \!\cdot\! 
\left(
    \lie \chi_1
     + \chi_2
     + \truth\chi_3
\right)
}
\end{align*}

where:
$$\chi_1=(\alpha-\!k_1-u) 
       \I\Ink{\!\!\alpha \!-\!\! k_1\!\!-\!\! u\! \geq\! 1\!\!}{1}{1} $$
$$\chi_2= \frac{\lie^2 (\beta\!-\!(\!d\!-\!k_1\!-\!k_2\!)-(x+z-u-v))}
            {\truth} $$
$$\chi_3= (\gamma-\!k_2-v) 
       \I\Ink{1}{1}{\!\!\gamma-\!k_2\!-\!v\!\geq\!1\!\!}$$

As $\truth\gg\lie$, to maximize the ratio, one need to maximize $\chi_1$ and minimize $\chi_3$. On the one hand $\chi_1$ is maximal when $\alpha$ is maximal (i.e. $\alpha=n$, as the maximum one can send to the same target with algorithm $\Abar_\Plie^{n,d}$ is $n$). On the other hand, $\chi_3$ is minimal when $\gamma=0$. Thus, the output producing the higher ratio is $\Output=\{n,d,0\}$.

To further increase the ratio, we need to minimize $k_1$ and $u$. The first one depends on $d$, a fixed parameter of the algorithm we cannot change. The later one, $u$, is varying from 0 to $x+z$ and takes its minimal value when $x+z=0$. We deduce from this that the two inputs producing the higher ratio are $\mathcal{G}=\{1,n-1,0\}$ and $\mathcal{G}'=\{0,n-1,1\}$.

By replacing the new indices in the ratio we obtain:

\begin{footnotesize}
$$ R_{n}^{d} \Inn{n}{d}{0} {0}{\!n\!-\!1\!}{0}
\!=\frac{  \sum_{k=0}^{d}{\binom{d}{k} \binom{n-1}{k} \truth^{k}\lie^{n-k-1} \Big( \truth+ k \cdot  \frac{\lie^{2}}{\truth} \Big)  }}{\sum_{k=0}^{d}{\binom{d}{k} \binom{n-1}{k} \truth^{k}\lie^{n-k-1} \Big( \frac{\sigma}{T-1}+ k \cdot \frac{\lie^{2}}{\truth}  \Big) }}\!=\!e^\epsilon $$
\end{footnotesize}
which leads to the result.

\end{proof}

\section{Details of Numerical Simulation}
\label{app:simu}
The results of our numerical simulation (Figure~\ref{fig:exp}) were obtained as follows.

First, given a number of target $T$, we fix an input computation graph with $n$ messages drawn from uniform and skewed target distributions. Then, two neighboring communication graphs $\mathcal{G}_1$ and $\mathcal{G}_2$ are generated by randomly changing the target of one message.

Second, for fixed parameters ($\Plie$, $d$ and $n$), we run our algorithm many times on the neighboring communication graphs  with different random seeds. 
For each run and each output $\mathcal{O}$ encountered, we count how many times $\mathcal{O}$ appeared for each communication graph:
\begin{align*}
    c_1(\mathcal{O}) &= \{\text{\# times }\mathcal{O}\text{ occurred when the input was }\mathcal{G}_1\},\\
    c_2(\mathcal{O}) &= \{\text{\# times }\mathcal{O}\text{ occurred when the input was }\mathcal{G}_2\},
\end{align*}
which correspond to the red and blue bars shown in Figure~\ref{fig:exp}.
We then compute an estimate of the log-ratio of the probabilities for each output $\mathcal{O}$ as:
$$r(\mathcal{O}) = \ln \Big(\max\Big\{\frac{c_1(\mathcal{O})}{c_2(\mathcal{O})}, \frac{c_2(\mathcal{O})}{c_1(\mathcal{O})}\Big\}\Big),$$
which correspond to the dotted green line in Figure~\ref{fig:exp}.

The total number of runs are set such that nearly all possible outputs are drawn at least once. To make sure that we did enough runs, we applied the Capture-recapture \cite{white1982capture} counting technique used in biology to estimate the size of populations of animals.\riad{une ref serait cool :)}
More precisely, we follow a two-step approach to estimate the percentage of unseen outputs. First, we divide the number of runs into different batches. We execute a batch and record all the different output drawn. Second, we draw another batch and count the number of new outputs (i.e., not seen in the first batch). The percentage of new outputs in the second batch represents the percentage of all possible outputs that have not been seen in any batch.  

\aurelien{On sait calculer le nombre total d'outputs possibles non? donc je vois pas tout à fait pourquoi on a besoin de cela.}

Interestingly, we observed that the choice of input has a negligible impact on the results: although the probability of individual outputs obviously depend on the input, the overall shape of the output probability distribution (and thus the results in Figure~\ref{fig:exp}) remains essentially the same.

\section{Technical Details and Proofs for the Results of Section~\ref{sec:amplification}}
\label{app:shuffling_proof}

In this section, we provide a detailed exposition of our analysis leading to Theorem~\ref{thm:main_amplification} in the main text.
We first introduce some key technical concepts in Section~\ref{app:key}. In Section~\ref{app:generic_sec}, to keep our analysis as general as possible (see Remark~\ref{rem:general}), we first prove $(\epsilon,\delta)$-DP results for a generic local randomizer $\Randomizer$. Finally, in Section~\ref{app:instantiate}, we apply our general result to our specific context.

For notational convenience, for any integer $n\geq 1$, we will denote the set $\{1,\dots,n\}$ by $[n]$.

\subsection{Key Concepts}
\label{app:key}

In this section, we review some key concepts from Balle et al. \cite{Balle2019} that we need to prove our results.

\paragraph{Decomposition of local randomizers.}
Let $\mathcal{X}$ and $\mathcal{Y}$ be some input and output domains respectively.
Let $\Randomizer:\mathcal{X}\rightarrow\mathcal{Y}$ be a \emph{local randomizer} taking an input $x\in\mathcal{X}$ and returning a randomized output $y\in\mathcal{Y}$.
The \emph{total variation similarity} $\gamma_\Randomizer$ of a local randomizer $\Randomizer$ measures the probability that $\Randomizer$ produces an output which is independent from its input. When this happens, the output is sampled from some distribution $\omega_\Randomizer$, which is called the \emph{blanket distribution} of $\Randomizer$. When it is clear from the context, we drop the subscript and simply write $\gamma$ and $\omega$.

We will leverage a decomposition of $\Randomizer$ as a mixture between an input-dependent and input-independent mechanism. Specifically, denoting by $\mu_x$ the output distribution of $\Randomizer(x)$, we write $\mu_x=(1-\gamma)\upsilon_x + \gamma\omega$. For a particular $\Randomizer:\mathcal{X}\rightarrow\mathcal{Y}$, the largest possible $\gamma$ is given by $\gamma=\int \inf_x \mu_x(y)dy$ and the corresponding blanket distribution $\omega$ is given by $\omega(y)=\inf_x \mu_x(y)/\gamma$. Balle et al. \cite{Balle2019} show that $\gamma\geq e^{-\epsilon_0}$ for any $\epsilon_0$-DP local randomizer, but it is possible to compute the exact value of $\gamma$ for common local randomizers (see Lemma~5.1 in \cite{Balle2019}).

We can illustrate these concepts on $\Randomizer_\Plie$, the local randomizer used in our algorithm. Since $\Randomizer_\Plie$ only randomizes the target, we can abstract away the source node and we have $\mathcal{X}=\mathcal{Y}=\Target$, $\gamma_{\Randomizer_\Plie}=\sigma$, $\upsilon_{\Randomizer_\Plie, t}(t') = \I[t = t']$ and $\omega_{\Randomizer_\Plie}(t') = 1/k$ for all $t'\in\Target$.

\paragraph{Hockey-stick divergence.} Differential privacy can be conveniently expressed as a divergence between distributions. Divergences come with known results and properties that provide useful technical tools to derive differential privacy guarantees. Below, we will use the characterization of  $(\epsilon,\delta)$-DP based on the so-called \emph{hockey-stick divergence}. Precisely, the hockey-stick divergence of order $e^\epsilon$ between distributions $\mu$ and $\mu'$ is defined as:
$$\Div_{e^\epsilon}(\mu||\mu') = \int [\mu(y)-e^\epsilon\mu'(y)]_+dy,$$
where $[\cdot]_+=\max(0,\cdot)$. The following lemma from \cite{Balle2019} shows the direct connection to $(\epsilon,\delta)$-DP.

\begin{lemma}
An algorithm $\mathcal{A}:\mathcal{X}^n\rightarrow \mathcal{Y}^m$ is $(\epsilon,\delta)$-DP if and only if $\Div_{e^\epsilon}(\mathcal{A}(\Dataset)||\mathcal{A}(\Dataset'))\leq \delta$ for any neigboring datasets $\Dataset=\{x_1,\dots,x_{n-1},x_n\}$ and $\Dataset'=\{x_1,\dots,x_{n-1},x'_n\}$.
\end{lemma}

\subsection{Privacy Guarantee for a Generic Local Randomizer}
\label{app:generic_sec}

Let $\epsilon,\epsilon_0\geq0$, $\delta\in(0,1)$ and $d\in\mathbb{N}$. In this section, we prove an $(\epsilon,\delta)$-DP result for algorithms $\mathcal{A}:\mathcal{X}^n\rightarrow \mathcal{Y}^{n+d}$ of the form $\mathcal{A}=\mathcal{S}_{\Randomizer,d}\circ \Randomizer^n$, where:
\begin{itemize}
    \item $\Randomizer^n:\mathcal{X}^n\rightarrow\mathcal{Y}^n$ is such that $$\Randomizer^n(x_1,\dots,x_n)=(\Randomizer(x_1),\dots,\Randomizer(x_n))$$
    where $\Randomizer:\mathcal{X}\rightarrow\mathcal{Y}$ is an arbitrary local randomizer satisfying $\epsilon_0$-DP.
    \item $\mathcal{S}_{\Randomizer,d}:\mathcal{Y}^n\rightarrow\mathcal{Y}^{n+d}$ randomly samples $d$ ``dummy'' messages from the blanket distribution $\omega_\Randomizer$ and shuffles (i.e., applies a random permutation to it) the multiset composed of the $n$ input messages and the $d$ dummy messages.
\end{itemize}
Note that $\mathcal{S}_{\Randomizer,0}$ (no dummy) corresponds to a standard shuffler: this is the setting covered by recent results on privacy amplification by shuffling, in particular those of \cite{Balle2019}.
The purpose of this section is to extend these results to account for the use of dummy messages, i.e., when $d>0$.

\paragraph{Step 1: Bounding the divergence in terms of i.i.d. random variables.}
Let $\Dataset=\{x_1,\dots,x_{n-1},x_n\}$ and $\Dataset'=\{x_1,\dots,x_{n-1},x'_n\}$ be two neighboring datasets that differ only in their last points $x_n$ and $x_n'$.
The key technical step of the proof is to bound the divergence $\Div_{e^\epsilon}(\mathcal{A}(\Dataset)||\mathcal{A}(\Dataset'))$ in terms of a sum of i.i.d. realizations of a ``privacy amplification'' random variable $L_\epsilon^{x_n,x_n'}$ defined as:
\begin{equation}
    \label{eq:amp_rv}
    L_\epsilon^{x_n,x_n'} = \frac{\mu_{x_n}(W) - e^\epsilon\mu_{x_n'}(W)}{\omega(W)},
\end{equation}
where $W\sim \omega$. We have the following result, which is the analog to Lemma~5.3 of \cite{Balle2019} for the case with dummies.

\begin{lemma}
\label{lem:2}
Let $\epsilon>0$ and let $\Dataset=\{x_1,\dots,x_{n-1},x_n\}$ and $\Dataset'=\{x_1,\dots,x_{n-1},x'_n\}$ two neighboring datasets with $x_n\neq x'_n$. We have:
\begin{align*}
&\Div_{e^\epsilon}(\mathcal{A}(\Dataset)||\mathcal{A}(\Dataset'))\\ 
&\leq
\sum_{m=0}^{n-1} \frac{1}{m+d+1}C^{n-1}_{m}\gamma^{m}(1-\gamma)^{n-1-m} \mathbb{E}\left[\sum_{i=1}^{m+d+1} L_i \right]_+\\
&= \frac{1}{\gamma n}\sum_{m=1}^{n} \frac{m}{m+d}C^{n}_{m}\gamma^{m}(1-\gamma)^{n-m} \mathbb{E}\left[\sum_{i=1}^{m+d} L_i \right]_+,
\end{align*}
where $L_1,\dots,L_{m+d}$ are i.i.d. copies of $L_\epsilon^{x_n,x'_n}$.
\end{lemma}

\begin{proof}
Recall that $\mathcal{A}=\mathcal{S}_{\Randomizer,d}\circ \Randomizer^n$. 
For a fixed input dataset $\Dataset=\{x_1,\dots,x_{n-1},x_n\}$, we define the random variable $Y_i\sim\mu_{x_i}$ for $i\in[n]$, where $\mu_x=(1-\gamma)\upsilon_x + \gamma\omega$ is the output distribution of $\Randomizer(x)$ as defined in Section~\ref{app:key}. For $i\in[d]$, we also define the random variable $Z_i\sim\omega$ . Using these notations, the output of $\mathcal{A}(\mathcal{D})$ can be seen as a realization of the random multiset $\mathcal{O}=\{Y_1,\dots,Y_n,Z_1,\dots,Z_d\}\in\mathbb{N}_{n+d}^\mathcal{Y}$, where $\mathbb{N}_{n+d}^\mathcal{Y}$ denotes the collection of all multisets of size $n+d$ of elements of $\mathcal{Y}$. Similarly, for a neighboring dataset $\Dataset'=\{x_1,\dots,x_{n-1},x'_n\}$, the output of $\mathcal{A}(\mathcal{D}')$ is a realization of the random multiset $\mathcal{O}'=\{Y_1,\dots,Y'_n,Z_1,\dots,Z_d\}$. Our goal is thus to bound $\Div_{e
^\epsilon}(\mathcal{O}||\mathcal{O}')$, where we use a slight abuse of notation by applying the divergence to random variables rather than distributions.

To exploit the mixture decomposition $\mu_x$, we define additional random variables. Let $V_i\sim\upsilon_{x_i}$ and $W_i\sim \omega$ for $i\in[n-1]$. Hence we have:
$$\def\arraystretch{1.2}
Y_i = \left\{ \begin{array}{ll} V_i & \text{with probability }1-\gamma,\\ W_i & \text{with probability }\gamma.\end{array} \right.$$
Finally, we define $\mathcal{B}\subseteq [n-1]$ to be the random subset of inputs among the first $n-1$ who sampled from the blanket, and let $\bar{\mathcal{B}} = [n-1] \setminus \mathcal{B}$. Note that for any $B\subseteq [n-1]$ we have $P[\mathcal{B}=B] = \gamma^{|B|}(1-\gamma)^{n-1-|B|}$.
With these notations, conditioned on a particular $B$, we have
$$\mathcal{O}|\{\mathcal{B}=B\} = \mathcal{W}_B \cup \mathcal{V}_{\bar{B}} \cup \mathcal{Z}_d \cup \{Y_n\},$$
where $\mathcal{W}_B = \{W_i\}_{i\in B}$, $\mathcal{V}_{\bar{B}} = \{V_i\}_{i\in[n-1]\setminus B}$ and $\mathcal{Z}_d = \{Z_i\}_{i=1}^d$.

By standard properties of the hockey-stick divergence (see Lemma A.1 in \cite{Balle2019}), we have:
\begin{equation}
\label{eq:divbound1}
\Div_{e^\epsilon}(\mathcal{O}||\mathcal{O}') \leq \sum_{B\subseteq [n-1]} \gamma^{|B|}(1-\gamma)^{n-1-|B|}\Div_{B,d}^{(1)}
\end{equation}
where $\Div_{B,d}^{(1)}=\Div_{e^\epsilon}(\mathcal{W}_B \cup \mathcal{V}_{\bar{B}} \cup \mathcal{Z}_d \cup \{Y_n\}||\mathcal{W}_B \cup \mathcal{V}_{\bar{B}} \cup \mathcal{Z}_d \cup \{Y'_n\})$.

By applying Lemma A.2 from \cite{Balle2019}, we further show that we can ignore the contributions of the first $n-1$ inputs who did not sample from the blanket. Precisely:
\begin{equation}
\label{eq:dbbound}
\Div_{B,d}^{(1)} \leq \Div_{e^\epsilon}(\mathcal{W}_B \cup \mathcal{Z}_d \cup \{Y_n\}||\mathcal{W}_B \cup \mathcal{Z}_d \cup \{Y'_n\}).
\end{equation}

Since the $W_i$'s are i.i.d., the distribution of $\mathcal{W}_B$ depends on $B$ only through its cardinality $m=|B|$. We thus define $\mathcal{W}_m=\{W_1,\dots,W_{m}\}$ for any $m\in[n-1]$, with $\mathcal{W}_0=\emptyset$. Rewriting \eqref{eq:divbound1} and \eqref{eq:dbbound}, we have shown that:
\begin{equation}
\label{eq:divbound2}
\Div_{e^\epsilon}(\mathcal{O}||\mathcal{O}') \leq \sum_{m=0}^{n-1} C^{n-1}_{m}\gamma^{m}(1-\gamma)^{n-1-m} \Div_{B,d}^{(2)}
\end{equation}
where $\Div_{B,d}^{(2)}=\Div_{e^\epsilon}(\mathcal{W}_m \cup \mathcal{Z}_d \cup \{Y_n\}||\mathcal{W}_m \cup \mathcal{Z}_d \cup \{Y'_n\})$.

We now upper bound the right-hand side of \eqref{eq:divbound2} in terms of the privacy amplification variables $L_1,\dots,L_{m+d}$ arising from the $m$ inputs who sampled from the blanket and the $d$ dummy messages.

Let $y\in\mathcal{Y}^{m+d}$ be a tuple of elements from $\mathcal{Y}$ and $Y\in\mathbb{N}_{m+d}^\mathcal{Y}$ be the corresponding multiset. We have:
\begin{align*}
& P[\mathcal{W}_{m-1} \cup \mathcal{Z}_d \cup \{Y_n\}=Y]\\
&= \frac{1}{(m+d)!}\sum_\tau P[(W_1,\dots,W_{m-1},Y_n,Z_1,\dots,Z_d)=y_\tau],    
\end{align*}
where $\tau$ ranges over all permutations of $\{1,\dots,m+d\}$ and we write $y_\tau=(y_{\tau(1)},\dots,y_{\tau(m+d)})$. Since $W_i\sim \omega$, $Y_n\sim \mu_{x_n}$ and $Z\sim \omega$, we have:
\begin{align*}
& P[(W_1,\dots,W_{m-1},Y_n,Z_1,\dots,Z_d)=y_\tau]\\
&=\omega(y_{\tau(1)})\dots\omega(y_{\tau(m-1)})\mu_{x_n}(y_{\tau(m)})\omega(y_{\tau(m+1)})\dots\omega(y_{\tau(m+d)}).
\end{align*}
Summing this expression over all permutations $\tau$ and factoring out $P[\mathcal{W}_{m} \cup \mathcal{Z}_d=Y]$ gives:
\begin{align*}
   &\frac{1}{(m+d)!} \sum_\tau \big(\omega(y_{\tau(1)})\dots\omega(y_{\tau(m-1)})\mu_{x_n}(y_{\tau(m)})\\
   &\quad\qquad\qquad\ \ \times\omega(y_{\tau(m+1)})\dots\omega(y_{\tau(m+d)})\\
   &\quad = \left(\prod_{i=1}^{m+d}\omega(y_i)\right)\frac{1}{m+d}\sum_{i=1}^{m+d} \frac{\mu_{x_n}(y_i)}{\omega(y_i)}\\
   &\quad = P[\mathcal{W}_m\cup \mathcal{Z}_d=Y]\frac{1}{m+d}\sum_{i=1}^{m+d} \frac{\mu_{x_n}(y_i)}{\omega(y_i)}.
\end{align*}

Plugging this in the definition of $\Div_{e^\epsilon}$, we get:
\begin{align*}
    &\Div_{e^\epsilon}(\mathcal{W}_{m-1} \cup \mathcal{Z}_d \cup \{Y_n\}||\mathcal{W}_{m-1} \cup \mathcal{Z}_d \cup \{Y'_n\})\\
    & = \int_{\mathbb{N}_{m}^\mathcal{Y}} \big[ P[\mathcal{W}_{m-1} \cup \mathcal{Z}_d \cup \{Y_n\}=Y]\\
    & \quad\qquad\ \  - e^\epsilon P[\mathcal{W}_{m-1} \cup \mathcal{Z}_d \cup \{Y'_n\}=Y] \big]_+ dY\\
    & = \int_{\mathbb{N}_{m}^\mathcal{Y}} P[\mathcal{W}_m\cup \mathcal{Z}_d=Y]\left[ \frac{1}{m+d}\sum_{i=1}^{m+d} \frac{\mu_{x_n}(y_i) - e^\epsilon\mu_{x'_n}(y_i)}{\omega(y_i)} \right]_+ dY\\
    & = \mathbb{E}\left[ \frac{1}{m+d}\sum_{i=1}^{m+d} \frac{\mu_{x_n}(y_i) - e^\epsilon\mu_{x'_n}(y_i)}{\omega(y_i)} \right]_+
    = \mathbb{E}\left[ \frac{1}{m+d}\sum_{i=1}^{m+d} L_i \right]_+.
\end{align*}

Plugging this into \eqref{eq:divbound2} completes the proof.
\end{proof}

\paragraph{Step 2: Bounding the sum of privacy amplification variables.}
To control the term $\mathbb{E}[\sum_{i=1}^{m+d} L_i ]_+$ in Lemma~\ref{lem:2}, we need to resort to concentration inequalities for sums of i.i.d. random variables.
We can trivially adapt Lemma 5.5 (based on Hoeffding's inequality) and Lemma 5.6 (based on Bennett's inequality) from \cite{Balle2019} by replacing $m$ by $m+d$.

\begin{lemma}
\label{lem:3}
Let $L_1,\dots,L_{m+d}$ be i.i.d. bounded random variables with $\mathbb{E}[L_i]=-a \leq 0$. Suppose that $b_- \leq L_i \leq b_+$ and let $b=b_+-b_-$. Then, by Hoeffding's inequality, the following holds:
$$\mathbb{E}\left[ \sum_{i=1}^{m+d} L_i \right]_+ \leq \frac{b^2}{4a}e^{-\frac{2(m+d)a^2}{b^2}}.$$
If furthermore we have $\mathbb{E}[L_i^2] \leq c$, then by Bennett's inequality:
$$\mathbb{E}\left[ \sum_{i=1}^{m+d} L_i \right]_+ \leq \frac{b_+}{a(m+d)\log(1+\frac{ab_+}{c})}e^{-\frac{(m+d)c}{b_+^2}\phi(\frac{ab_+}{c})},$$
where $\phi(u) = (1+u)\log(1+u) - u$.
\end{lemma}

Bounding the sum of privacy amplification random variables with the above lemma requires the knowledge of their expected value as well as bounds on the values they can take. 
Bennett's inequality further requires a bound on the second moment: this can lead to a tighter (albeit more complex) bound compared to Hoeffding's.
Balle et al. \cite{Balle2019} provide such bounds that hold for any local randomizer $\Randomizer$ satisfying $\epsilon_0$-DP, as we recall below.

\begin{lemma}
\label{lem:4}
Let $\Randomizer:\mathcal{X}\rightarrow\mathcal{Y}$ be an $\epsilon_0$-DP local randomizer with total variation similarity $\gamma$. For any $\epsilon\geq 0$ and any $x,x'\in\mathcal{X}$, the privacy amplification variable $L=L_\epsilon^{x,x'}$ satisfies the following:
\begin{enumerate}
    \item $\mathbb{E}[L] = 1-e^\epsilon$,
    \item $\gamma e^{-\epsilon_0}(1-e^{\epsilon+2\epsilon_0}) \leq L \leq \gamma e^{\epsilon_0}(1-e^{\epsilon-2\epsilon_0})$,
    \item $\mathbb{E}[L^2] \leq \gamma e^{\epsilon_0} (e^{2\epsilon} + 1) - 2\gamma^2e^{\epsilon-2\epsilon_0}$.
\end{enumerate}
\end{lemma}

\paragraph{Step 3: Putting everything together.}
Based on the intermediate results above, we can obtain an $(\epsilon,\delta)$-DP guarantee that holds for any $\epsilon_0$-DP local randomizer. We illustrate this using the simpler Hoeffding's inequality in Lemma~\ref{lem:3}.

\begin{theorem}
Let $\Randomizer:\mathcal{X}\rightarrow\mathcal{Y}$ be an $\epsilon_0$-DP local randomizer with total variation similarity $\gamma$, and $d\geq 0$. The algorithm $\mathcal{A}=\mathcal{S}_{\Randomizer,d}\circ\Randomizer^n$ is $(\epsilon,\delta)$-DP for any $\epsilon$ and $\delta$ satisfying
$$\frac{1}{\gamma n}\sum_{m=1}^{n} \frac{m}{m+d}C^{n}_{m}\gamma^{m}(1-\gamma)^{n-m} \frac{b^2}{4a}e^{-\frac{2(m+d)a^2}{b^2}}\leq \delta,$$
where $a=1-e^\epsilon$ and $b = \gamma(1+e^\epsilon)(e^{\epsilon_0} - e^{-\epsilon_0})$.
\end{theorem}

Although the above theorem does not give a simple expression for $\epsilon$ and $\delta$, it can be easily evaluated numerically, for instance to identify the lowest achievable $\epsilon$ given the $\epsilon_0$ and $\gamma$ of $\Randomizer$ and the desired $\delta$.




\subsection{Application to our Setting}
\label{app:instantiate}

We can now apply our previous privacy guarantees, which hold for an arbitrary randomizer, to our specific context. Since our local randomizer $\Randomizer_\Plie$ only randomizes the target, for simplicity of notations we abstract away the source node and consider that $\Randomizer_\Plie$ operates on $\mathcal{X}=\Target$ and returns an element of $\mathcal{Y}=\Target$. We have $\gamma_{\Randomizer_\Plie}=\sigma$, $\upsilon_{\Randomizer_\Plie, t}(t') = \I[t = t']$ and $\omega_{\Randomizer_\Plie}(t') = 1/k$ for all $t'\in\Target$.
$\Randomizer_\Plie$ is in fact equivalent to $T$-ary randomized response \cite{kRR}, and we can refine the results of Lemma~\ref{lem:4}. This is shown in the following lemma, adapted from \cite{Balle2019}.
\begin{lemma}
\label{lem:5}
For any $\epsilon>0$ and $t,t'\in\Target$, the privacy amplification variable $L=L_\epsilon^{t,t'}$ satisfies the following:
\begin{itemize}
    \item $-(1-\sigma)Te^\epsilon + \sigma(1-e^\epsilon) \leq L \leq (1-\sigma)T - \sigma(1-e^\epsilon)$,
    \item $\mathbb{E}[L^2] = \sigma(2-\sigma)(1-e^\epsilon)^2+(1-\sigma)^2T(1+e^{2\epsilon})$.
\end{itemize}
\end{lemma}

We can use the results of Lemma~\ref{lem:5} instead of the generic ones from Lemma~\ref{lem:4} to obtain tighter privacy guarantees that are specific to randomized response. Doing this with Hoeffding's inequality and combining with Lemma~\ref{lem:2} gives the statement for $\sigma>0$ in Theorem~\ref{thm:main_amplification}.

\paragraph{Special case where source nodes do not sample ($\sigma=0$).}
Interestingly, when $\sigma=0$ (i.e., source nodes never sample and always send their true message to the scrambler), the analysis is still valid and we can get $(\epsilon,\delta)$-DP guarantees that rely entirely on the dummies added by the scrambler. In this case, the set $B$ in the proof of Lemma~\ref{lem:2} (the set of source nodes who sample from the blanket) is always empty and therefore the result of Lemma~\ref{lem:2} simplifies to:
$$
\Div_{e^\epsilon}(\mathcal{A}(\Dataset)||\mathcal{A}(\Dataset')) \leq
\frac{1}{d+1} \mathbb{E}\left[\sum_{i=1}^{d+1} L_i \right]_+.
$$

The bounds in Lemma~\ref{lem:5} also simplify when $\sigma=0$, namely: $-Te^\epsilon\leq L\leq T$ and $\mathbb{E}[L^2] = T(1+e^{2\epsilon})$. 
We can thus apply Hoeffding's or Bennett's inequalities with the above values to bound the sum of i.i.d. variables as in Lemma~\ref{lem:3} using $d+1$ instead of $m+d$, and get privacy guarantees for this special case as well. Using Hoeffding's gives the statement for $\sigma=0$ in Theorem~\ref{thm:main_amplification}.

\end{document}